\documentclass[journal]{IEEEtran}
\usepackage{stfloats}
\usepackage{algorithm}  
\usepackage{algpseudocode}  
\usepackage{amsmath}

\usepackage{graphicx}
\usepackage[]{algpseudocode}
\usepackage{algorithmicx,algorithm}
\usepackage{amsmath}
\usepackage{titlesec}
\usepackage{multirow}  % 多行合并的宏包
\usepackage{amsthm,amsmath,amssymb}
\usepackage{graphicx}
\usepackage{float}
\usepackage{subfigure}
\usepackage{mathrsfs}
\usepackage{url}
\usepackage{array}
\newcolumntype{C}[1]{>{\centering}p{#1}}
\usepackage{makecell}
\usepackage{diagbox}
\usepackage{xcolor}
\usepackage{soul} % 导入 soul 包

\usepackage{gensymb}
\usepackage{cite}
\usepackage{color}

\ifCLASSOPTIONcompsoc
 \usepackage[caption=false,font=normalsize,labelfont=sf,textfont=sf]{subfig}
\else
 \usepackage[caption=false,font=footnotesize]{subfig}
\fi

\usepackage{amsthm}

\UseRawInputEncoding
\begin{document}
% \theoremstyle{plain} 
% \newcommand{\remark}[1]{\textmd{#1}}
% \usepackage{hyperref}
% 如果不采用章节号做前缀，则不用[section]
\theoremstyle{definition}
\newtheorem{remark}{Remark}

% \swapnumbers

\newtheorem{theorem}{Theorem}
% 这句定义使得definition环境和theorem共享编号
\newtheorem{definition}[theorem]{Definition} 
\newtheorem{lemma}{Lemma} 
% 这句定义使得corollary环境和theorem共享编号
\newtheorem{corollary}[theorem]{Corollary}
\newtheorem{example}{Example}[section]
\newtheorem{proposition}[theorem]{Proposition}

% \linespread{1.67}

% \newpage
\vspace{-15mm}
\title{Electromagnetic Property Sensing in ISAC with Multiple Base Stations: Algorithm, Pilot Design, and Performance Analysis % Performance Analysis
}

% \author{Yuhua Jiang and Feifei Gao    % <-this 

\author{   
Yuhua Jiang, Feifei Gao, Shi Jin, and Tie Jun Cui

% \thanks{Manuscript received September 11, 2022; revised October 17, 2022; accepted December 13, 2022. The editor who approves its publication is Dajana Cassioli.
% This work was supported in part by the Beijing Natural Science Foundation under Grant L222002, by the National Natural Science Foundation of China under Grant 61831013, by Tsinghua University Initiative Scientific Research Program, and by the Grant from the Guoqiang Institute (2021GQG1022, Tsinghua University), (Corresponding author: Feifei Gao). 
% }
\thanks{
Y. Jiang and F. Gao are with Institute for Artificial Intelligence, Tsinghua University (THUAI), 
State Key Lab of Intelligent Technologies and Systems, Tsinghua University, 
Beijing National Research Center for Information Science and Technology (BNRist), Beijing, P.R. China (email: jiangyh221600@gmail.com, feifeigao@ieee.org).

% Miaomiao Dong is with the Huawei Theoretical Research Department, Hong Kong (e-mail: dong.828599@huawei.com)

% Y. Liu is with the Intelligent Sensing Laboratory,
% Department of Electronic Engineering, Tsinghua University, Beijing 100084,
% China (e-mail: yiminliu@tsinghua.edu.cn).

S. Jin is with the National Mobile Communications Research 
Laboratory, Southeast University, Nanjing 210096, China (e-mail: jinshi@seu.edu.cn).

T. J. Cui is with the State Key Laboratory of Millimeter Waves, Southeast University, Nanjing 210096, China (e-mail: tjcui@seu.edu.cn).

% W. Zhang is with the School of Electrical Engineering and Telecommunications, University of New South Wales, Sydney, NSW 2052, Australia (e-mail: w.zhang@unsw.edu.au).
% M. Jian is with the State Key Laboratory of Mobile Network and Mobile Multimedia Technology, Shenzhen 518055, China, and also with Algorithm Dept., ZTE Corporation, Shenzhen 518057, China.
% S. Zhang is with the State Key % Laboratory of Integrated Services Networks, Xidian University, Xi’an 710071,
% P. R. China (e-mail: zhangshunsdu@xidian.edu.cn).
}
}

\maketitle
\vspace{-15mm}
\begin{abstract}

% However, most existing literature about near-field imaging applies discrete aperture RIS, which can not achieve high spatial resolution and energy efficiency.
Integrated sensing and communication (ISAC) has opened up numerous game-changing opportunities for future wireless systems. 
In this paper, we develop a novel scheme that utilizes orthogonal frequency division multiplexing (OFDM) pilot signals to sense the electromagnetic (EM) property of the target and thus identify the materials of the target.
Specifically, we first establish an EM wave propagation model with Maxwell equations, where the EM property of the target is captured by a closed-form expression of the channel. 
We then build the mathematical model for the relative permittivity and conductivity distribution (RPCD) within a predetermined region of interest shared by multiple base stations (BSs). 
\textcolor{blue}{%
By leveraging the Lippmann-Schwinger equation, we propose an EM property sensing method that reconstructs the RPCD using compressive sensing techniques. This approach exploits the joint sparsity of the EM property vector, which enables the proposed method to effectively handle the high dimensionality and ill-posed nature of the inverse scattering problem.}
We then develop a fusion algorithm to combine data from multiple BSs, which can enhance the reconstruction accuracy of EM property by efficiently integrating diverse measurements. 
Moreover, the fusion is performed at the feature level of RPCD and features low transmission overhead. 
We further design the pilot signals that can minimize the mutual coherence of the equivalent channels and enhance the diversity of incident EM wave patterns. 
Simulation results demonstrate the efficacy of the proposed method in achieving high-quality RPCD reconstruction and accurate material classification. 
% Moreover, the RPCD reconstruction quality and material classification accuracy can be improved with increased signal-to-noise ratio (SNR) or the number of BSs. 
\end{abstract}

% Integrated sensing and communication (ISAC) has opened up numerous game-changing opportunities for future wireless systems. In this paper, we develop a novel scheme that utilizes orthogonal frequency division multiplexing (OFDM) pilot signals in ISAC systems to sense the electromagnetic (EM) property of the target and then identify the material of the target. Specifically, we first establish an EM propagation model by means of Maxwell equations, where the EM property of the target is captured by a closed-form expression of the ISAC channel. We then model the relative permittivity and conductivity distribution (RPCD) within a specified detection region. Based on the sensing model, we introduce a multi-frequency-based EM property sensing method by which the RPCD can be reconstructed from compressive sensing techniques that exploits the joint sparsity structure of the EM property vector. To improve the sensing accuracy, we design a beamforming strategy from the communications transmitter based on the Born approximation that can minimize the mutual coherence of the sensing matrix. Simulation results demonstrate the efficacy of the proposed method in achieving high-quality RPCD reconstruction and accurate material classification. Furthermore, improvements in RPCD reconstruction quality and material classification accuracy are observed with increased signal-to-noise ratio (SNR) or reduced target-transmitter distance.

% massive MIMO,

\IEEEpeerreviewmaketitle

% often require complex optical systems and specialized components, which can be costly and challenging to manufacture and maintain.  
% Evaluating a larger electromagnetic spectrum is one way to master these challenges. In this study, the infrared (IR) emissivity as a material specific property is investigated regarding its suitability for increasing the material classification reliability. 

\section{Introduction} 
% Electromagnetic (EM) property is a key indicator in identifying the material of targets. 
% Material identification plays a pivotal role in various industries, enabling automation processes to operate efficiently and accurately \cite{nature1}.   

Recently, the emergence of integrated sensing and communication (ISAC) has unlocked a multitude of revolutionary possibilities for the forthcoming wireless sensing systems. 
ISAC facilitates the sharing of the limited radio spectrum between sensing and communications (S\&C) systems and thereby significantly reduces resource costs \cite{isac1,mypaper3,isac2}. 
In contrast to the conventional approach of separate functionalities for S\&C, ISAC methodology brings two distinct advantages.
Firstly, by enabling the joint utilization of constrained resources such as spectrum, energy, and hardware platforms, ISAC enhances efficiency for both S\&C and thus yields integration benefits. 
Secondly, the collaborative synergy between S\&C fosters coordination gains and thereby amplifies the overall performance. 
Consequently, ISAC has the potential to improve the capabilities of both S\&C functionalities in the foreseeable future.
Thanks to its myriad benefits, ISAC would play a pivotal role in various forthcoming applications, including intelligent connected vehicles, internet of things (IoT), and smart homes and cities \cite{isac4,isac6,isac10}.

% \color{blue}
% ISAC with multiple base stations (BSs) has emerged as a pivotal area of research, which focuses on the coordinated use of resources across a network of BSs to simultaneously enhance S\&C capabilities \cite{co5,co6,co3}. 
% This approach is underpinned by the principle of leveraging shared information and coordinated actions among multiple BSs. 
% Recent advancements have explored various facets of ISAC with multiple BSs, including power control \cite{co1}, cooperative sensing \cite{co2}, and beamforming strategies \cite{co4}, which collectively optimize the BSs' ability to sense the environment in a more accurate way while maintaining robust communication links.
% These efforts are crucial to realize the next generation of wireless networks, where the seamless integration of S\&C will lead to innovative applications in areas such as autonomous vehicles, smart cities, and advanced surveillance systems \cite{co7,co8,co9,co10}.
% \color{black}

\color{blue}
ISAC with multiple base stations (BSs) has emerged as a crucial area 
of research, which focuses on the coordinated use of resources across a network of BSs to simultaneously enhance S\&C capabilities \cite{co5,co6,co3}. 
This approach leverages shared information and coordinated actions among BSs, optimizing performance through advancements in symbol-level cooperative sensing \cite{wei2023symbol}, joint fronthaul compression and beamforming design \cite{zhang2024optimal}, and UAV trajectory design for coordinated beamforming \cite{cheng2024networked}.
Shared information among BSs collectively improves environmental sensing accuracy at the cost of large transmission overhead between BSs and the central processing unit. The cooperation of multiple BSs is vital for the next-generation networks, enabling applications like autonomous vehicles, smart cities, and advanced surveillance systems, where ISAC technologies offer both high-resolution sensing and low-latency communications \cite{co7,co8,co9, co10}. %
\color{black}

Actually, ISAC technologies primarily serve the purpose of generating an accurate representation that connects the physical world with its communication digital twin counterpart \cite{twin,twin2}. 
Unlike image-based digital twins that focus on shape and location, communication oriented digital twins require reconstructing the communication channel as well as the channel related tasks. 
In this context, accurate electromagnetic (EM) property of the objects in the physical world is vital to build communication digital twins, because it determines the propagation direction, attenuation, and physical phenomena such as diffraction, reflection, and scattering of EM waves that transmit communication signals.
Meanwhile, the property of EM characteristics is essential for the identification of materials in targets and would significantly impact a wide range of industries.
For instance, sensing EM property is applicable in the inspection and imaging of the human body, which caters to the growing needs in social security, environmental surveillance, and medical fields.

Traditional methods sense EM property via infrared detection, which utilizes the unique infrared emissivity property of each material \cite{nature1}. 
However, these methods typically require expensive and specialized equipment that is difficult to produce and sustain \cite{infrared1}. 
Additionally, infrared detection is hindered by its limited ability to penetrate targets, especially when the surface is covered with camouflage, which makes the conventional methods ineffective. 
\color{blue}
The low-frequency EM wave is a promising alternative for reliable EM property sensing because they can penetrate various objects more effectively than infrared waves \cite{material_identification1}.
Due to their longer wavelengths, low-frequency EM waves interact less with small particles and surface irregularities, reducing absorption and scattering.
Unlike infrared detection, low-frequency EM waves in communication band enable practical sensing by leveraging the channel state information (CSI) embedded in transmitted signals, which allows ISAC technologies to identify EM property reliably without the need for specialized equipment.
\color{black}

% \textcolor{blue}{% 
% However, these methods typically require expensive and specialized equipment that is difficult to produce and sustain \cite{infrared1}. 
% Additionally, infrared detection is hindered by its limited ability to penetrate targets, especially when the surface is covered with camouflage, which makes the conventional methods ineffective. 
% On the other hand, the ability of relatively low-frequency EM waves to penetrate various objects makes EM wave detection in communication frequency band a practical and promising alternative for reliable EM property sensing \cite{material_identification1}. 
% Since the EM property of target materials is implicitly embedded within the channel state information (CSI) from transmitters to receivers, 
% it is possible to sense EM property via ISAC technologies. }%

Although ISAC has demonstrated significant achievements in localization \cite{isac8}, tracking \cite{isac3}, imaging \cite{RIS_image}, as well as various other applications \cite{isac1,isac2,mypaper}, EM property sensing within the ISAC framework has been barely studied. 
\textcolor{blue}{% 
To the best of the authors' knowledge, only one work \cite{mypaper4} implements EM property sensing under single-BS ISAC systems. 
However, the estimation accuracy and the resolution under a single-BS ISAC system are not satisfactory, and a single BS may not cover a large sensing region. 
Although multi-BS ISAC systems offer significant potential to improve sensing accuracy by combining observations from different locations, they face challenges in efficiently fusing distributed data with limited transmission overhead. Additionally, \cite{mypaper4} only discusses the reduction of pilot mutual coherence within the same subcarrier. 
Reducing pilot mutual coherence across subcarriers remains unexplored and could enhance sensing reliability and spectrum efficiency in large-scale systems.
}% 

In this paper, we develop a novel scheme that utilizes pilot signals in an ISAC system with multiple BSs to sense the EM property of the target and identify the constituent materials. 
The main contributions are summarized as follows:
\begin{itemize} 
\item[$\bullet$] 
We rigorously derive a formulation of EM property sensing equations from the EM scattering theory for multiple user equipments (UEs) and multiple BSs. 
This approach particularly focuses on the accurate estimation of relative permittivity and conductivity distribution (RPCD) in a region of interest shared by all BSs.
\item[$\bullet$] 
\textcolor{blue}{We propose a fusion algorithm to combine data from multiple BSs at the RPCD feature level, which reduces the error of RPCD reconstruction by efficiently integrating signals on different subcarriers and boasts relatively low transmission overhead.}%
\item[$\bullet$]  We design orthogonal frequency division multiplexing (OFDM) pilots across multiple subcarriers to reduce mutual coherence in ISAC systems, which guarantees the feasibility of employing compressive sensing techniques to estimate EM property and thereby secures the reliability and accuracy of the overall sensing process. 
\end{itemize} 
Simulation results reveal that the proposed method successfully delivers high-quality reconstruction of RPCD as well as precise material classification. Moreover, improvements in the quality of RPCD reconstruction and the accuracy of material classification can be achieved by either increasing the signal-to-noise ratio (SNR) at the receivers or utilizing the data from more BSs.

% EM property sensing formulation in ISAC with multiple BSs:
% Joint multi-subcarrier pilot design to minimize mutual coherence:
% Feature-level multi-BS data fusion: 

% Next, we utilize RIS to generate different virtual EM masks on the target surface and calculate the cross-correlation between the mask patterns and the electric field strength at the receiver.
% We then provide a RIS design scheme for virtual EM masks by employing a regularization technique.
% The cross-range resolution of the proposed method is analyzed based on the spatial spectrum of the generated masks.

% We then establish an end-to-end EM propagation model and provide a RIS design criterion to generate virtual EM masks. 

% Section~\uppercase\expandafter{\romannumeral3} derives the EM Wave propagation formulas. 

The rest of this paper is organized as follows. 
Section~\uppercase\expandafter{\romannumeral2} presents the system model and formulates the EM property sensing problem.
Section~\uppercase\expandafter{\romannumeral3} elaborates the strategy of sensing EM property. 
Section~\uppercase\expandafter{\romannumeral4} describes the scheme of data fusion with multiple BSs.
Section~\uppercase\expandafter{\romannumeral5} proposes the approach to designing the pilots and identifying the materials of the target.
Section~\uppercase\expandafter{\romannumeral6} provides the numerical simulation results, and 
Section~\uppercase\expandafter{\romannumeral7} draws the conclusion.

Notations: Boldface denotes a vector or a matrix;  $j$ corresponds to the imaginary unit; $(\cdot)^H$, $(\cdot)^T$, and $(\cdot)^*$ represent Hermitian, transpose, and conjugate, respectively; 
$\circ$ denotes the Khatri-Rao product; 
$\odot$ denotes the Hadamard product;
$\otimes$ denotes the Kronecker product;
$\mathrm{vec}(\cdot)$ denotes the vectorization operation; 
$\textrm{diag}(\mathbf{a})$ denotes the diagonal matrix whose diagonal elements are the elements of $\mathbf{a}$; 
$\mathbf{I}$, $\mathbf{1}$, and $\mathbf{0}$ denote the identity matrix, the all-ones vector, and the all-zeros vector with compatible dimensions;   
$\left\Vert\mathbf{a}\right\Vert_2$ denotes $\ell2$-norm of the vector $\mathbf{a}$; 
$\left|\mathbf{a}\right|$ denotes the 
element-wise absolute value of the complex vector $\mathbf{a}$;
$\Re(\cdot)$ and $\Im(\cdot)$ denote the real and imaginary part of complex vectors or matrices, respectively; 
$\text{Tr}(\cdot)$ denotes the trace of a matrix; 
$\textbf{prox}_{f(\cdot)}$ denotes the proximal mapping of the function $f(\cdot)$; 
$\left\Vert\mathbf{A}\right\Vert_F$ denotes Frobenius-norm of the matrix $\mathbf{A}$; 
$\mathbf{A} \succeq \mathbf{0}$ indicates that the matrix $\mathbf{A}$ is positive semi-definite; 
the distribution of a circularly symmetric complex Gaussian (CSCG) random vector with zero mean and 
covariance matrix $\mathbf{A}$ is denoted as $\mathcal{C N} (\mathbf{0}, \mathbf{A})$.

\section{System Model}
 
\subsection{Signal Model}
\begin{figure}[t]
  \centering  \centerline{\includegraphics[width=8.9cm]{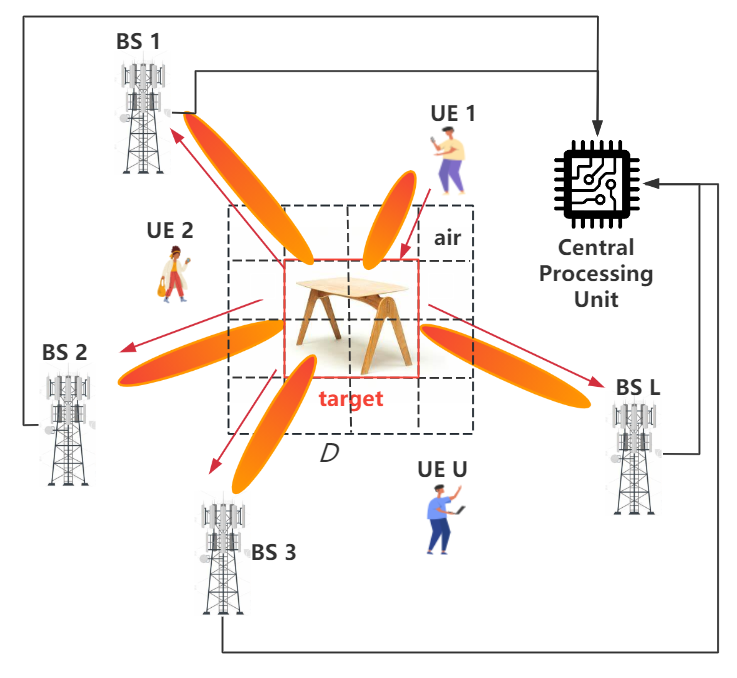}}
  \caption{The system model is composed of $U$ multi-antenna UEs, a target to be sensed in region $D$, and $L$ multi-antenna BSs. 
  % The target is located in the near field of the transmitter.
  } 
  \label{system_model} %
\end{figure}

% In order to yield high-quality sensing performance, we assume that the target is located in the near field of the transmitter. 
Consider an uplink OFDM communication scenario with $U$ mobile UEs and $L$ BSs.
The $L$ BSs collaboratively provide EM property sensing service for a single target and are connected using cables by a central processing unit (CPU).
Assume that each UE is equipped with $N_t$ transmitting antennas, and each BS is equipped with $N_r$ receiving antennas. 
The scenario is also known as the multi-static ISAC system.
The UEs all adopt fully digital precoding structure where the number of RF chains $N_{RF}$ is equal to the number of antennas $N_t$. 
To facilitate EM property sensing, we utilize pilot transmissions over $K$ subcarriers. The central frequency of the OFDM signals is represented by $f_c$, and the interval between subcarriers is denoted by $\Delta f$.

% Consequently, the total signal bandwidth is defined as $W = (K-1)\Delta f$. The frequency at the $k$-th subcarrier can be calculated using the formula $f_k = f_c + (2k-K-1)\Delta f/2$, where $k$ ranges from 1 to $K$. The wavelength corresponding to the $k$-th subcarrier is given by $\lambda_k = \frac{c}{f_c + (2k-K-1)\Delta f/2}$. 

In order to avoid interference between multiple UEs, we assume the ISAC system utilizes the time-division duplex (TDD) mode, i.e., only one UE sends the pilot signals at one time. 
Additionally, it is assumed that each UE transmits $I$ pilot symbols on each subcarrier. 
% Thus, we call the above multi-static MIMO system as a single-transmit multi-receive multi-subcarrier (ST-MR-MS) MIMO system as shown in Fig.~1. 

% \textcolor{black}{ 
% The comprehensive sensing of a target in ISAC is accomplished through the following steps.
% Firstly, the existence of the target is ascertained through target detection. This initial phase is focused on discovering the presence of the target, which is a fundamental prerequisite for any further analysis.
% Secondly, once the target is detected, the task of ISAC is pinpointing the precise location and velocity of the target.
% These parameters are crucial as they provide the spatial context needed to accurately engage with the target. 
% The methods to acquire these parameters have been comprehensively discussed in conventional ISAC literature \cite{isac1,isac2,isac3,isac4,mypaper,isac5,isac6,isac7}.
% Thirdly, we proceed to sense the EM property of the target, acquiring the material-specific information.    
% Therefore, when sensing the EM property of the target, we may assume that other property like the position of the target have already been obtained and accurately known. 
% } 

The overall sensing procedure involves several essential steps.
Initially, the target's presence is detected \cite{detection,isac5,isac6,isac7}, which is a critical prerequisite for subsequent analysis. 
Next, ISAC determines the exact location and velocity of the detected target \cite{isac1,isac2,isac3,isac4,mypaper}.
These parameters offer crucial spatial context, facilitating accurate interaction with the target. 
Lastly, the EM property of the target is sensed to gather material-specific information \cite{mypaper4}. 
Hence, we here assume that information such as the approximate region occupied by the target has already been obtained.

Suppose the target scatters the pilot signals from the UE to the BSs as shown in Fig.~\ref{system_model}. 
\color{blue}
Denote the target scattering path channel from the $u$-th UE to the $l$-th BS on the $k$-th subcarrier as $\hat{\mathbf{H}}_{k,l,u} \in \mathbb{C}^{N_r \times  N_t}$ and the direct path channel from the $u$-th UE to the $l$-th BS on the $k$-th subcarrier as $\hat{\mathbf{H}}_{k,l,u}^{d} \in \mathbb{C}^{N_r \times  N_t}$.
Then, the signals received by the $l$-th BS can be formulated as
\begin{align}
\hat{\mathbf{y}}_{k,l,u} = \hat{\mathbf{H}}_{k,l,u} \mathbf{w}_{k,u} x_{k,u} + \hat{\mathbf{H}}_{k,l,u}^{d} \mathbf{w}_{k,u} x_{k,u}
+ \hat{\mathbf{n}}_{k,l}, 
\end{align}
where $\mathbf{w}_{k,u} \in \mathbb{C}^{N_t \times  1}$ is the digital beamforming vector on the $k$-th subcarrier, and 
$\hat{\mathbf{n}}_{k,l} \sim \mathcal{C N}\left(\mathbf{0}_{N_r}, \sigma_{k,l}^2 \mathbf{I}_{N_r}\right)$ is the complex Gaussian noise at the $l$-th BS on the $k$-th subcarrier. 
The pilot symbol $x_{k,u}$ can be set as $1$ without loss of generality. 
Since only the signals scattered by the target carry the information of its EM property, the BS may purposely extract the received signals arriving from the target's direction by beamforming. 
Assume prior knowledge has determined that the target is contained within a certain region to be sensed, represented by $D$.
Denote \( [\theta_l^{min}, \theta_l^{max}] \) as the angular spread of region $D$ when observed from the $l$-th BS. 
Define the steering vector \( \mathbf{a}_k(\theta_l) \in \mathbb{C}^{N_r \times 1} \) for \( \theta_l \) as 
$\mathbf{a}(\theta_l) = \frac{1}{\sqrt{ N_r }} \!\! \begin{bmatrix}
1 , e^{-j \frac{2\pi}{\lambda_k} d_l \sin(\theta_l)} , e^{-j \frac{2\pi}{\lambda_k} 2d_l \sin(\theta_l)} ,\!\! \cdots \!\! , e^{-j \frac{2\pi}{\lambda_k} (N_r-1)d_l \sin(\theta_l)} \!\!\!
\end{bmatrix}^\top\!\!\!\!,$
where \( d_l \) is the inter-antenna spacing.
To extract the signals from the target scattering path, we apply receiver beamforming at the BS with the following beamforming matrix:
\begin{align}
 \mathbf{P}_{k,l} = \frac{\int_{\theta_l^{min}}^{\theta_l^{max}} \mathbf{a}_k(\theta_l) \mathbf{a}_k(\theta_l)^H \mathrm{d} \theta_l}{\theta_l^{max}-{\theta_l^{min}}}  .
\end{align}

% Since the angle of arrival for the direct path of the UE is usually not included in $[\theta_l^{min}, \theta_l^{max}]$
If the $u$-th UE, the target, and the $l$-th BS are not on the same line, then there is $\mathbf{P}_{k,l} \hat{\mathbf{H}}_{k,l,u}^d \mathbf{w}_{k,u} \approx \mathbf{0} $, especially when the number of the antennas at the BS is large as is true in the coming 6G era. 
Thus, the received signals after beamforming effectively isolate the direct path of the $u$-th UE and can be simplified to
\begin{align}
\tilde{\mathbf{y}}_{k,l,u} \overset{\Delta}{=} \mathbf{P}_{k,l} \hat{\mathbf{y}}_{k,l,u} =\mathbf{P}_{k,l} \hat{\mathbf{H}}_{k,l,u} \mathbf{w}_{k,u} + \tilde{\mathbf{n}}_{k,l},
\end{align}
where $ \tilde{\mathbf{n}}_{k,l} \overset{\Delta}{=} \mathbf{P}_{k,l} \hat{\mathbf{n}}_{k,l} $.
% This demonstrates that receiver beamforming using the steering vector \( \mathbf{a}(\theta) \) projects the received signal onto the target's direction, thereby extracting the desired signal while mitigating other components.
If the $u$-th UE, the target, and the $l$-th BS happen to be on the same line, then the other UEs can be used to sense the target. 
% Thus, the power of $\hat{\mathbf{H}}_{k,l,u}^{d} \mathbf{w}_{k,u}$ is much smaller than the power of $\hat{\mathbf{H}}_{k,l,u} \mathbf{w}_{k,u}$. 
% We consequently denote the received signals after receiver beamforming as: 
% \begin{align}
% \tilde{\mathbf{y}}_{k,l,u} = \hat{\mathbf{H}}_{k,l,u} \mathbf{w}_{k,u}
% + \tilde{\mathbf{n}}_{k,l} .
% \end{align}
\color{black}
\begin{remark}
We will demonstrate that $\mathbf{w}_{k,u}$ should be designed in a varying way across the subcarriers and hence a fully digital system is required. 
\end{remark}

Suppose we uniformly discretize region $D$ into $M$ sampling points.
Let the vectors $\mathbf{E}_k^{i,D} \in \mathbb{C}^{M \times 1}$ and $\mathbf{J}_k^{D} \in \mathbb{C}^{M \times 1}$ collect the incident electric field and the equivalent contrast source at all $M$ sampling points within $D$, respectively \cite{lipp,m-born,v-born}. 
Define $\mathbf{H}_{1,k,u} \in \mathbb{C}^{M \times N_t}$ as the channel matrix from the $u$-th
UE to the target that maps $\mathbf{w}_{k,u}$ to $\mathbf{E}_k^{i,D}$, and define $\mathbf{H}_{2,k,l} \in \mathbb{C}^{N_r \times M}$ as the channel matrix from the target to the $l$-th BS path that maps $\mathbf{J}_k^{D}$ to $\tilde{\mathbf{y}}_{k,l}$.
The ISAC channel from the UE through region $D$ to the $l$-th BS can then be described as:
\begin{align}
\hat{\mathbf{H}}_{k,l,u} = \mathbf{H}_{2,k,l} \mathbf{X}_k \mathbf{H}_{1,k,u}, 
\label{X} 
\end{align}  
where $\mathbf{X}_k \in \mathbb{C}^{M \times M}$ maps $\mathbf{E}_k^{i,D}$ to $\mathbf{J}_k^{D}$ and is irrelevant to $u$ and $l$, which 
represents the influence of the existence of the target on the signal transmission. 
Thus, the material-related EM property of the target is implicitly incorporated into the formulation of $\mathbf{X}_k$, which will be leveraged to identify the material of the target. 

\textcolor{blue}{%
The matrices $\mathbf{H}_{1,k,u}$ and $\mathbf{H}_{2,k,l}$ are also known as the radiation operator matrices 
and are solely related to the physical laws of EM propagation \cite{operator,lipp}. 
Since we assume that the positions of the receiver, the region of interest $D$, and the transmitter are known, 
$\mathbf{H}_{1,k,u}$ and $\mathbf{H}_{2,k,l}$ can be calculated by the computational EM methods 
such as method of moments (MoM) \cite{lipp}, \cite{MoM}, finite element method (FEM) \cite{FEM}, and finite-difference time-domain (FDTD) \cite{FDTD}. 
In this paper, we use MoM to compute the radiation operator matrices, which is also widely used in ground penetrating radar \cite{ebihara2007mom,alkhalifeh2016efficient,ebihara2004mom}. 
}%

% In this paper, our objective is then to retrieve the EM property of the target using the received signals $\tilde{\mathbf{y}}_{k,l,u}$ and to further identify the material of the target from $\mathbf{X}_k$.  

% \begin{align}
% \mathbf{H}_{1,k}[i,j] &= g_t k_{k}^{2}  \int_{T_j} {G_{k}({\mathbf {r}_i,\mathbf{r}}') } d{\mathbf {r}}', 
% \label{X2} \\
% \mathbf{H}_{2,k}[i,j] &= g_r k_{k}^{2}  \int_{D_j} {G_{k}({\tilde{\mathbf {r}}_i,\mathbf{r}}) } d{\mathbf {r}}, 
% \label{X2}
% \end{align}
% where $g_r$ and $g_t$ denote the antenna gains of the receiver and the transmitter,
% $T_j$ and $D_j$ denote the regions of the $j$-th transmitting antenna and the $j$-th discretized cell within $D$, 
% $\mathbf {r}_i$ and $\tilde{\mathbf {r}}_i$ denote the positions of the $i$-th discretized cell within $D$ and the $i$-th receiving antenna, 
% $k_k=\frac{2\pi}{\lambda_k}$ denotes the wavenumber at the $k$-th subcarrier,  
% and ${G_{k}({\mathbf {r},\mathbf{r}}')}$ denotes the Green's function.

% \textcolor{red}{
% The location of the target is implicitly incorporated in the formulation of $\mathbf{H}_{1,k}$ and $\mathbf{H}_{2,k}$, which can be leveraged to estimate the location of the target by the received signals in conventional ISAC literature \cite{isac1,isac2,isac3,mypaper,isac5,isac6,isac7}.
% However, it is important to note that in our context, the target's location is a prerequisite, particularly for sensing specific EM property of the target.
% }
% \subsection{Problem Formulation}

\subsection{EM Wave Propagation Model}
In this subsection, the closed-form expression of $\mathbf{X}_k$ is derived to unveil the influence of the target scattering process on the signal transmission. 

Sensing the EM property in region $D$ involves reconstructing the contrast function $\chi_k(\mathbf{r}^{\prime})$, which is defined as 
the difference in complex relative permittivity between the target and air. 
Given that the relative permittivity and conductivity of air are approximately equal to $1$ and $0$ S/m, respectively, we can formulate $\chi_k(\mathbf{r}^{\prime})$ as \cite{born,born2,m-born} 
\begin{align}
\chi_k(\mathbf{r}^{\prime}) =  \epsilon_r(\mathbf{r}^{\prime}) + \frac{j \sigma(\mathbf{r}^{\prime})}{\epsilon_0 \omega_k} - 1,
\label{def000}  
\end{align} 
where $\epsilon_r(\mathbf{r}^{\prime})$ denotes the real relative permittivity at point $\mathbf{r}^{\prime}$, $\sigma(\mathbf{r}^{\prime})$ denotes the conductivity at point $\mathbf{r}^{\prime}$, $\omega_k=2\pi f_k$ denotes the angular frequency of EM waves at the $k$-th subcarrier, and  $\epsilon_0$ is the vacuum permittivity. 
We then aim to reconstruct the distributions of the relative permittivity, $\epsilon_r(\mathbf{r}')$, and the conductivity, $\sigma(\mathbf{r}')$. 
The total electric field for the $k$-th subcarrier within region $D$ can be described by the Lippmann-Schwinger (LS) equation \cite{lipp, lipp2}
\begin{align}
E_{k}^{t} ({\mathbf {r}}) = E_{k}^{i} ({\mathbf {r}})+k_{k}^{2}  \int  \limits _{D} {G_{k}({\mathbf {r},\mathbf{r}}') } \chi_{k} ({\mathbf {r}}') E_{k}^{t} ({ \mathbf {r}}')d{\mathbf {r}}' , 
\label{lipp1}%, as detailed in
\end{align}
where $ k_k = \frac{\omega_k}{c} $ represents the wavenumber at the $k$-th subcarrier.
Moreover, 
$E_{k}^{t} ({\mathbf {r}})$,       
$E_{k}^{i} ({\mathbf {r}})$, 
and ${G_{k}({\mathbf {r},\mathbf{r}}')}$ denote the total electric field, the incident electric field, and the Green's function, respectively. 

% Consider a 2D transverse magnetic (TM) scenario, and ${G_{k}({\mathbf {r},\mathbf{r}}')}$ can then be formulated as \cite{lipp}, \cite{lipp2}
% \begin{equation}
% {G_{k}({\mathbf {r},\mathbf{r}}')}=\frac{j}{4} H_0^{(2)}\left(k_k\left\Vert\mathbf{r}-\mathbf{r}^{\prime}\right\Vert_2\right) ,
% \end{equation}
% where $H_0^{(2)}$ is the zero-order Hankel function of the second kind.

% Additionally, it is assumed that the contrast function $\chi_k(\mathbf{r}^{\prime})$ is constant within each element in  $D$ and is $0$ outside $D$.

Let the vectors $\mathbf{E}_k^{s,D} \in \mathbb{C}^{M \times 1}$, $\mathbf{E}_k^{t,D} \in \mathbb{C}^{M \times 1}$, and $\boldsymbol{\chi}_k \in \mathbb{C}^{M \times 1}$ collect the scattering electric field, the total electric field, and $\chi_k(\mathbf{r})$ on all $M$ sampling points within $D$, respectively.
Let \(\mathbf{G}_k \in \mathbb{C}^{M \times M}\) represent the matrix obtained by discretizing the integral kernel \(k_k^2 G_k(\mathbf{r} , \mathbf{r}')\) in (\ref{lipp1}) using the method of moments (MoM). In $\mathbf{G}_k$, both \(\mathbf{r}\) and \(\mathbf{r}'\) correspond to the discretized sampling points within region \(D\).

% The number of discrete points on the RIS is denoted as $N_R$, and the total electric field at each antenna location at the receiving end is represented as $\mathbf{E}_R^{t} \in \mathbb{C}^{N_R}$. Then, equation (\ref{lipp}) can be represented in matrix form as:

% Where $\mathbf{E}_R^{i} \in \mathbb{C}^{N_R}$ and $\mathbf{E}_R^{s} \in \mathbb{C}^{N_R}$ represent the electric field strengths of the LOS and NLOS components at the receiving end, respectively.
% Here, $\mathbf{r}^{\prime}_m$ represents the coordinates of the $m$-th discrete point within the target.

\color{blue}
The total electric field within $D$ can be obtained by adding the incident field and the scattering field, which means $\mathbf{E}_k^{t,D}$ can be written in the discretized form of (\ref{lipp1}) specialized with $\mathbf{r}\in D$ as:
\begin{equation}
\mathbf{E}_k^{t,D}=\mathbf{E}_k^{i,D}+\mathbf{E}_k^{s,D}
=\mathbf{E}_k^{i,D}+ \mathbf{G}_k \textrm{diag}\left ({\boldsymbol{\chi }_{k}}\right) \mathbf{E}_k^{t,D} .
\label{lip3}
\end{equation}

In order to yield a closed-form expression of $\mathbf{E}_k^{t,D}$, 
we reformulate (\ref{lip3}) in the nonlinear form with respect to $\textrm {diag}\left ({\boldsymbol{\chi }_{k}}\right)$ and then obtain $\mathbf{J}_k^{D}$ as:  
\begin{align}
&\mathbf{E}_k^{t,D}  =(\mathbf{I}-\mathbf{G}_k \textrm {diag}\left ({\boldsymbol{\chi }_{k}}\right))^{-1} \mathbf{E}_k^{i,D} , \label{lip3.5} \\
&\mathbf{J}_k^{D}  = \textrm {diag}\left ({\boldsymbol{\chi}_{k}}\right)  \mathbf{E}_k^{t,D} = \textrm {diag}\left ({\boldsymbol{\chi}_{k}}\right) 
(\mathbf{I}-\mathbf{G}_k \textrm {diag}\left ({\boldsymbol{\chi }_{k}}\right))^{-1} \mathbf{E}_k^{i,D} . 
\label{lip3.6}
\end{align}
\color{black}

According to (\ref{X}) and (\ref{lip3.6}), the closed-form expression of $\mathbf{X}_k$ in (\ref{X}) can be written as:
\begin{align}
\mathbf{X}_k =  
\textrm {diag}\left ({\boldsymbol{\chi}_{k}}\right)
\left[\mathbf{I}-\mathbf{G}_k \textrm {diag} \left ({\boldsymbol{\chi }_{k}}\right)\right]^{-1} .  
\label{X_1}
\end{align}
For those sampling points occupied by the air within $D$, the corresponding rows of $\mathbf{X}_k$ are all zeros, and those points are not regarded as EM sources of scattering waves.

\textcolor{blue}{Let $\mathbf{w}_{k,u,i}, i =1, \cdots, I$, denote $\mathbf{w}_{k,u}$ of (1) in the $i$-th time slot.}
Denote $\tilde{\mathbf{W}}_{k,u} =\left[\mathbf{w}_{k,u,1},  \mathbf{w}_{k,u,2}, \cdots, \mathbf{w}_{k,u,I}\right] \in \mathbb{C}^{N_t \times I}$ as the digital transmitter beamforming matrix stacked by time on the $k$-th subcarrier.
Then the overall received pilot signals $\mathbf{Y}_{k,l,u}\in \mathbb{C}^{N_r \times I}$ can be formulated in a compact form as:
\begin{align}
\mathbf{Y}_{k,l,u} &=  \left [\tilde{\mathbf{y}}_{k,l,u,1},\tilde{\mathbf{y}}_{k,l,u,2},\cdots,\tilde{\mathbf{y}}_{k,l,u,I} \right]
\nonumber \\ 
& =   \mathbf{P}_{k,l} \mathbf{H}_{2,k,l}
\textrm {diag}
\left ({\boldsymbol{\chi}_{k}}\right)
[\mathbf{I}-\mathbf{G}_k \textrm {diag}\left (\boldsymbol{\chi }_{k}\right)]^{-1} 
\mathbf{H}_{1,k,u} \tilde{\mathbf{W}}_{k,u}
 \nonumber \\
& + \left[\tilde{\mathbf{n}}_{k,l,1}, \tilde{\mathbf{n}}_{k,l,2}, \cdots, \tilde{\mathbf{n}}_{k,l,I}\right]. 
\label{y1}
\end{align} 

Define the equivalent channel from the $u$-th user to the target as $\tilde{\mathbf{H}}_{1,k,u} = \mathbf{H}_{1,k,u} \tilde{\mathbf{W}}_{k,u}$ and horizontally concatenate the equivalent channels into $\bar{\mathbf{H}}_{1,k} = \left[\tilde{\mathbf{H}}_{1,k,1},\cdots,\tilde{\mathbf{H}}_{1,k,U} \right]$.
Then we can vectorize $\mathbf{Y}_{k,l,u}$ $(u=1,\cdots,U)$ to $\mathbf{y}_{k,l} \in \mathbb{C}^{U I N_r \times  1}$ as:
\begin{align}
& \mathbf{y}_{k,l} = \mathrm{vec}\left([\mathbf{Y}_{k,l,1},\cdots,\mathbf{Y}_{k,l,U}]\right) \nonumber \\ 
& \overset{(a)}{=}  \!\!
\left[\left(\bar{\mathbf{H}}_{1,k}^\top
\left[\mathbf{I}-\mathbf{G}_k \textrm {diag}\left ({\boldsymbol{\chi }_{k}}\right)\right]^{-\top}\right) \circ  (\mathbf{P}_{k,l} \mathbf{H}_{2,k,l}) \right] \boldsymbol{\chi }_{k} + 
\mathbf{n}_{k,l} ,
\label{n1}
\end{align} 
where $\overset{(a)}{=}$ in (\ref{n1}) comes from the equality: $\mathrm{vec}\left(\mathbf{A} \textrm{diag}(\mathbf{b}) \mathbf{C}\right)=\left(\mathbf{C}^T \circ \mathbf{A}\right) \mathbf{b}$, and $\mathbf{n}_{k,l}=\left[\tilde{\mathbf{n}}_{k,l,1}^\top, \tilde{\mathbf{n}}_{k,l,2}^\top, \cdots, \tilde{\mathbf{n}}_{k,l,I}^\top\right]^\top$ denotes the overall vectorized noise.

In order to derive a quasi-linear sensing model, we further define the sensing matrix $\mathbf{D}_{k,l} \in \mathbb{C}^{U I N_r \times M  }$ as: 
\begin{align}
\mathbf{D}_{k,l} &\overset{\Delta}{=} \left\{\bar{\mathbf{H}}_{1,k}^\top
\left[\mathbf{I}-\mathbf{G}_k \textrm {diag}\left ({\boldsymbol{\chi }_{k}}\right)\right]^{-\top}\right\} \circ (\mathbf{P}_{k,l} \mathbf{H}_{2,k,l}) .   
\label{sense1}
\end{align}
Then the sensing equation with respect to $\boldsymbol{\chi}_k$ can be formally formulated as: 
\begin{align}
\mathbf{y}_{k,l} & = \mathbf{D}_{k,l} \boldsymbol{\chi}_k
+\mathbf{n}_{k,l} .
\label{sensing}
\end{align}

\color{blue}
\begin{remark} 
Note that the existing multistatic radar system differs significantly from the system presented in this paper.
In radar systems, signals are both transmitted and received by radars, whereas in this paper, mobile UEs transmit the signals and the fixed BSs receive the signals.  
UEs and BSs vary a lot in their numbers, sizes of antenna arrays, and distances to the target. 
A lot of UEs surrounding the target can provide more views than only using the BSs. 
The transmit pilot design is also tailored for UEs with relatively small size of antenna arrays and with relatively close distance to the target. 
Additionally, traditional multistatic radar imaging systems focus solely on target shape, without considering the distribution of dielectric constants and conductivity. The system most similar to that in this paper is the multi-antenna EM imaging technology \cite{lipp,m-born,lipp2}. However, in the multi-antenna EM imaging technology, densely arranged antennas are placed surrounding the target, with each antenna sequentially transmitting signals, and all received signals are directly uploaded for centralized processing. Compared to multi-antenna EM imaging, the key differences in this paper are: (1) the signals received by the antennas are not directly uploaded to the CPU but instead, the sensing results are fused to reduce communication overhead; (2) not every antenna transmits signals sequentially, but rather, beamforming is applied at the transmitter to reduce the mutual coherence of the pilot patterns, thereby improving sensing accuracy.
\end{remark} 
\color{black}

\section{EM Property Sensing} 
In this section, we propose an EM property sensing method by fusing received pilot signals of all $K$ subcarriers from the same BS.
Denote the central angular frequency as $\omega_c = 2\pi f_c$ and denote the corresponding contrast vector as $\boldsymbol{\chi}_c$.
Denote $\mathbf{z}_{k,l} = \left[ \Re\left(\mathbf{y}_{k,l}\right)^\top , \Im\left(\mathbf{y}_{k,l}\right)^\top
\right]^\top$ and the EM property vector as $\mathbf{s}= \left[(\boldsymbol{\varepsilon}_r-\mathbf{1}_{M})^\top,(\frac{1}{\omega_c \boldsymbol{\varepsilon}_0}\boldsymbol{\sigma})^\top \right]^\top$. We then define the real-valued weighted sensing matrix as:
\begin{align}
\mathbf{E}_{k,l} = \left[\!\!\begin{array}{cc}
\Re\left(\mathbf{D}_{k,l}\right) & - \frac{\omega_c}{\omega_k}\Im\left(\mathbf{D}_{k,l}\right) \\
\Im\left(\mathbf{D}_{k,l}\right) & \frac{\omega_c}{\omega_k}\Re\left(\mathbf{D}_{k,l}\right)
\end{array}\!\!\right]\! .
\label{ee}
\end{align}
Separating the real and imaginary part of equation (\ref{sensing}) and extracting frequency-related components in $\boldsymbol{\chi}_k$, we can derive the real-valued sensing equation as:
\begin{align}
\mathbf{z}_{k,l}=\mathbf{E}_{k,l} \mathbf{s} + \left[\begin{array}{l}
\Re\left(\mathbf{n}_{k,l}\right)^\top,
\Im\left(\mathbf{n}_{k,l}\right)^\top
\end{array}\right]^\top.
\label{Ek2}
\end{align}
Define the general vector of measurement $\tilde{\mathbf{z}}_l \in \mathbb{R}^{ 2 U I N_r K \times 1}$ and the general sensing matrix $\tilde{\mathbf{E}}_l \in \mathbb{R}^{ 2U I N_r K \times 2M}$ as:
\begin{align}
\tilde{\mathbf{z}}_l &= \left[\mathbf{z}_{1,l}^\top, \cdots, \mathbf{z}_{K,l}^\top \right]^\top ,  \\  
\tilde{\mathbf{E}}_l &= \left[\mathbf{E}_{1,l}^\top, \cdots, \mathbf{E}_{K,l}^\top \right]^\top.  \label{assemble}
\end{align} 
\textcolor{blue}{Real-world objects with irregular shapes often exhibit sparsity when placed within artificially defined and regularly shaped regions.\footnote{\textcolor{blue}{Since the BSs are unaware of the target's accurate size, the BSs need to choose a large enough $D$ to guarantee that the target with the possibly largest size can be contained by $D$.
Besides, the region $D$ is in squared shape while the target's shape is often irregular. 
Moreover, if there is no high requirement for the detailed shape of the reconstructed target, the range of $D$ can be expanded to include the surrounding air, thereby maintaining the sparse structure on $\mathbf{s}$ at the cost of reducing the resolution of the target's shape.
Thus, some part of the region $D$ is filled with air under most circumstances, which makes the sparse structure on $\mathbf{s}$ holds true.
}}}
Given that the region $D$ is mainly filled with air and the target occupies a relatively small part within $D$, most of the elements of $\mathbf{s}$ are $0$ corresponding to the EM property of the air. 
\textcolor{blue}{%
Additionally, $(\boldsymbol{\varepsilon}_r-\mathbf{1}_{M})$ and $\frac{1}{\omega_c \boldsymbol{\varepsilon}_0}\boldsymbol{\sigma}$ 
share the same support set and the same sparse structure. 
When a target occupies $\tilde{K}$ sampling points in the region $D$, both  $(\boldsymbol{\varepsilon}_r-\mathbf{1})$ and $\frac{1}{\omega_c \boldsymbol{\varepsilon}_0}\boldsymbol{\sigma}$ will have $\tilde{K}$ non-zero elements in the same positions. Consequently, the vector $\mathbf{s}$, which combines $(\boldsymbol{\varepsilon}_r-\mathbf{1}_{M})$ and $\frac{1}{\omega_c \boldsymbol{\varepsilon}_0}\boldsymbol{\sigma}$, will be jointly $\tilde{K}$-sparse.
}%
Utilizing the compressive sensing techniques, $\mathbf{s}$ can be reconstructed 
by the $l$-th BS individually 
as the solution to the following regularized mixed $\ell(1,2)$-norm minimization problem: 
\begin{subequations}
\begin{align}
\min_\mathbf{s} \  & h_l (\mathbf{s}) \overset{\Delta}{=} \frac{1}{2} \| \tilde{\mathbf{z}}_l- \tilde{\mathbf{E}}_l \mathbf{s}\|_2^2 + \lambda_l  \|\mathbf{s}\|_{1,2} \\
\text { s.t. }
 & \mathbf{s} \geq \mathbf{0}_{2M} ,   
\end{align}
\label{com}%
\end{subequations}
where $\lambda_l$ represents a predetermined positive regularization parameter that should be chosen appropriately. 
The constraint $\mathbf{s} \geq \mathbf{0}_{2M}$ indicates all elements of $\mathbf{s}$ are nonnegative and holds true because the conductivity is nonnegative and the relative permittivity is larger than or equal to $1$. 
Moreover, $\|\mathbf{s}\|_{1,2}$ is the mixed $\ell(1,2)$-norm defined as:
$
\| \mathbf{s}\|_{1,2}
\overset{\Delta}{=}
{\sum _{m=1}^{M}
\sqrt{s_m^2+s_{m+M}^2}} .    
$

% \begin{equation} 
% \| \mathbf{s}\|_{1,2}
% \overset{\Delta}{=}
% {\sum _{m=1}^{M}
% \sqrt{s_m^2+s_{m+M}^2}} .    
% \end{equation}

% Consequently, there is $\|\mathbf{s}\|_{1} = \mathbf{1}_{M}^\top \mathbf{s}$

In particular, (\ref{com}) 
can be considered as the maximum a posteriori (MAP) estimation problem,
where the prior probability density function of $\mathbf{s}$ is $\lambda_l e^{-\lambda_l \|\mathbf{s}\|_{1,2}}$. 
Given the prior information, (\ref{com}) focuses on harnessing the joint sparsity of $\mathbf{s}$ to enhance the sensing performance. 

\section{Feature-Level Multi-BS Data Fusion} 
In this section, we design the multi-BS data fusion algorithm to improve the accuracy of RPCD reconstructions compared to the single-BS sensing results.
In order to reduce the transmission overhead between BSs and the CPU, we perform the feature-level fusion based on the EM property of the collected data. 
\textcolor{blue}{%
In the feature-level fusion, only the sensing results of RPCD are transmitted between the BSs and the CPU, rather than the raw data, i.e., the received signals. 
}%
Thus, the feature-level fusion reduces the amount of information exchanged while preserving the quality and utility of the data for further processing and analysis.

\subsection{Multi-Agent Consensus Equilibrium}
\textcolor{blue}{%
In the multi-BS network, each BS collects information from an individual perspective and can be considered as an agent for the shared EM property sensing task.
Since only the sensing results of RPCD are transmitted between the BSs and the CPU, the CPU is unaware of the sensing equations (\ref{com}).
The BSs need to solve (\ref{com}) distributedly and the CPU needs to achieve consensus among the sensing results of different BSs.
Therefore, we utilize the multi-agent consensus equilibrium (MACE) framework to harmonize the sensing results of multiple BSs \cite{mace1}, because it efficiently addresses the challenge of fusing the distributed data with low transmission overhead.
}% 

In order to solve problem (\ref{com}) from signals received at all BSs, we introduce the general MAP cost function $h(\mathbf{s})$ as the sum of $h_l(\mathbf{s})$, aiming to minimize these individual cost functions collaboratively, i.e., % with a common $\mathbf{s}$
\begin{subequations}
\begin{align}
\min_\mathbf{ s } \  &  h(\mathbf{s}) \overset{\Delta}{=} \sum_{l=1}^L \zeta_l h_l(\mathbf{s}) \\ 
\text { s.t. } 
 & \mathbf{s} \geq \mathbf{0}_{2M} ,  
\end{align} 
\label{gen0} 
\end{subequations}
where $\zeta_l$ is the normalized weight coefficient and satisfies $\sum_{l=1}^L \zeta_l = 1 $.

% In the context of EM property sensing, the $l$-th BS plays the role of an agent, denoted by $F_l(\cdot)$, which maintains an individual copy of RPCD $\mathbf{s}_l$, and $F_l(\mathbf{s}_l)$ is
% an improved reconstruction according to the $l$-th agent.

In order to understand the relationship between (\ref{gen0}) and MACE, 
we define the $l$-th agent $F_l(\cdot)$ as the proximal mapping of $\sigma^2 \zeta_l h_l(\cdot)$, i.e., a mapping in the following form:   % \operatorname
\begin{align} 
F_l(\mathbf{s}) \!& =\! \textbf{prox}_{ \sigma^2 \zeta_l h_l(\cdot) } \left(\mathbf{s}\right) 
\! = \! \underset{\mathbf{v}, \mathbf{v} \ge \mathbf{0}_{2M}}{\arg\min}\left\{\frac{\|\mathbf{v}-\mathbf{s}\|_2^2}{2 } + \sigma^2 \zeta_l h_l(\mathbf{v})\right\} \nonumber \\
& = \underset{\mathbf{v}, \mathbf{v} \ge \mathbf{0}_{2M}}{\arg\min}\left\{\frac{\|\mathbf{v}-\mathbf{s}\|_2^2}{2 \sigma^2 \zeta_l} + \frac{1}{2} \| \tilde{\mathbf{z}}_l-\tilde{\mathbf{E}}_l \mathbf{v}\|_2^2 + \lambda_l  \|\mathbf{v}\|_{1,2}\right\} , 
\label{prox}%
\end{align}
where $\sigma$ is a hyperparameter shared by all BSs that controls the convergence speed. 

The core concept of MACE is that each BS applies an abstract force, represented as an offset, to a consensus solution.
MACE is achieved when the collective forces exerted by all the BSs balance each other out, which is called equilibrium in physics \cite{mace2}. 
Specifically, the consensus solution of all BSs, defined as $\mathbf{s}^* \in \mathbb{R}^{ 2 M \times 1}$, is found by solving the following set of MACE equations:
\begin{align}
F_l\left(\mathbf{s}^*+\mathbf{u}_l^*\right) & =\mathbf{s}^* , \quad l = 1, \cdots, L, \label{def0} \\ 
\sum_{l=1}^L \mathbf{u}_l^* & = \mathbf{0}_{2M} ,
\label{sum0}
\end{align}
where $\mathbf{u}_l^* \in \mathbb{R}^{2 M \times 1}$ represents the abstract force applied by the $l$-th BS that should satisfy (\ref{sum0}) in order to achieve equilibrium among all BSs. 

\color{blue}
\begin{theorem}
The solution $\mathbf{s}^*$ to the MACE equations (\ref{def0}) and (\ref{sum0})
is exactly the solution to the MAP reconstruction problem  (\ref{gen0}).
\end{theorem}

\begin{proof}
According to (\ref{prox}) and (\ref{def0}) as well as the convexity of $h_l(\cdot)$, there is 
\begin{align} 
\left[ \mathbf{s}^*-\left(\mathbf{s}^*+\mathbf{u}_l^*\right) \right] + \sigma^2 \zeta_l \partial h_l\left(\mathbf{s}^*\right)  \ni \mathbf{0}_{2M}, \quad l = 1, \cdots, L, 
\label{stationary}
\end{align}
where $\partial h_l(\cdot)$ is the subgradient of $h_l(\cdot)$. Thus, by summing (\ref{stationary}) over $l$ and applying (\ref{sum0}), we obtain % with respect to $\mathbf{s}$
\begin{align} 
\partial h \left(\mathbf{s}^*\right) \ni \mathbf{0}_{2M}. 
\label{optim}
\end{align}
Note that (\ref{optim}) always holds irrespective of the choice of $\sigma$, which only influences the time to reach consensus. 
Considering the convexity of $h(\cdot)$, we have proven that the stationary point  $\mathbf{s}^*$ is the global minimum of the general MAP cost function $h(\cdot)$ in (\ref{gen0}). 
\end{proof}
\color{black}

The main calculation step in the MACE equations is to compute $F_l(\cdot)$ in (\ref{prox}).
In order to solve (\ref{prox}) efficiently, we define 
the indicator function corresponding to the constraint $\mathbf{v} \ge \mathbf{0}_{2M}$ as
\begin{equation}
\mathbb{I}(\mathbf{\mathbf{v}})= \begin{cases}0, & \mathbf{v} \ge \mathbf{0}_{2M} \\ +\infty, &  \text{otherwise} \end{cases} . 
\label{indicator}
\end{equation}
Then the constraint $\mathbf{v} \ge \mathbf{0}_{2M}$ can be replaced by incorporating the indicator function (\ref{indicator}) into the objective function. 
We can then solve (\ref{prox}) using the alternating direction method of multipliers (ADMM). With an auxiliary variable $\mathbf{a} \in \mathbb{R}^{2 M\times 1}$, (\ref{prox}) can be equivalently transformed into 
\begin{subequations}
\begin{align}
\min_{\mathbf{v}, \mathbf{a}} \  &  \frac{1}{2} \| \tilde{\mathbf{z}}_l-\tilde{\mathbf{E}}_l \mathbf{v}\|_2^2 + \mathbb{I}(\mathbf{\mathbf{a}}) + \lambda_l  \|\mathbf{a}\|_{1,2} +  \frac{1}{2 \sigma^2 \zeta_l} \|\mathbf{a}-\mathbf{s}\|_2^2 \label{ssum}\\ 
\text { s.t. }
 & \mathbf{v} = \mathbf{a}.
\end{align}
\label{admm}%
\end{subequations}
The reason for the above transformation is that the first term in (\ref{ssum}) can be regarded as the fidelity term, while the remaining three terms in (\ref{ssum}) can be regarded as the regularization terms given by prior knowledge.
Since both the fidelity term and the regularization terms are convex functions, ADMM guarantees the convergence to the global minimum of the objective function (\ref{ssum}) \cite{boyd2004convex}.
To equivalently convert (\ref{admm}) into an optimization problem without the constraint, the augmented Lagrangian function $L_l(\mathbf{v}, \mathbf{a}, \mathbf{b})$ 
associated with the optimization problem 
(\ref{admm}) is defined as
\begin{align}
L_l(\mathbf{v}, \mathbf{a}, \mathbf{b}) & =\frac{1}{2} \| \tilde{\mathbf{z}}_l-\tilde{\mathbf{E}}_l \mathbf{v}\|_2^2 + \mathbb{I}(\mathbf{\mathbf{a}}) + \lambda_l  \|\mathbf{a}\|_{1,2} +  \frac{\|\mathbf{a}-\mathbf{s}\|_2^2}{2 \sigma^2 \zeta_l} 
\nonumber\\
&-\textbf{b}^{\top}(\mathbf{a}-\mathbf{v})+\frac{1}{2 \eta_l}\|\mathbf{a}-\mathbf{v} \|_2^2 ,
\end{align}
where $\mathbf{b} \in \mathbb{R}^{2 M\times 1} $ is the Lagrangian multiplier and $\eta_l$ is the penalty factor.
Given the initial vector set $\left\{\mathbf{v}^{(0)}, \mathbf{a}^{(0)},\mathbf{b}^{(0)} \right\}$ and until a certain stopping criterion is satisfied, ADMM generates the sequence $(\mathbf{v}^{(n)}, \mathbf{a}^{(n)}, \mathbf{b}^{(n)})$ recursively as follows
\begin{align}
&\mathbf{v}^{(n+1)}  = \underset{\mathbf{v}}{\arg \min } \, L_l\left(\mathbf{v}, \mathbf{a}^{(n)}, \mathbf{b}^{(n)}\right) \nonumber\\
&=\left(\tilde{\mathbf{E}}_l^{\top} \tilde{\mathbf{E}}_l + \frac{1}{\eta_l} \mathbf{I}\right)^{-1}\left[\tilde{\mathbf{E}}_l^{\top} \tilde{\mathbf{z}}_l + \frac{1}{\eta_l}\left(\mathbf{a}^{(n)} - \eta_l \mathbf{b}^{(n)} \right)\right],  \label{viter} \\
&\mathbf{a}^{(n+1)}  = \underset{\mathbf{a}}{\arg \min } \, L_l \left(\mathbf{v}^{(n+1)}, \mathbf{a}, \mathbf{b}^{(n)}\right) \!= \! \underset{\mathbf{a}}{\arg \min } \! \left\{ \mathbb{I}(\mathbf{\mathbf{a}}) \right. \! \nonumber\\
& + \left. \!\lambda_l  \|\mathbf{a}\|_{1,2} \!+\!  \frac{\|\mathbf{a}-\mathbf{s}\|_2^2}{2 \sigma^2 \zeta_l} \!+\!\frac{\|\mathbf{a}- \eta_l \mathbf{b}^{(n)}-\mathbf{v}^{(n+1)} \|_2^2}{2 \eta_l} \! \right\} \!\! , \label{aaaaa} \\
&\mathbf{b}^{(n+1)}  = \mathbf{b}^{(n)}-\frac{1}{\eta_l} \left(\mathbf{v}^{(n+1)}-\mathbf{a}^{(n+1)}\right) . \label{bbbbb}
\end{align}
Since (\ref{aaaaa}) is relatively difficult to solve, we then elaborate the details of the update of $\mathbf{a}^{(n+1)}$. 
Let $a^{(n+1)}_m$ denote the $m$-th element of $\mathbf{a}^{(n+1)}$.
Note that (\ref{aaaaa}) can be decomposed into $M$ uncorrelated subproblems, with the $m$-th subproblem being formulated as
\begin{align}
 &(a^{(n+1)}_m,a^{(n+1)}_{m+M}) \!\! = \!\! \underset{a_{m},a_{m+M}}{\arg \min } \!\! \left \{ \mathbb{I}(a_m,a_{m+M}) \! + \! \lambda_l \sqrt{a_{m}^2 + a_{m+M}^2} \right. \nonumber\\
\! & \! + \!(\frac{1}{2 \eta_l} \!+ \!\frac{1}{2 \sigma^2 \zeta_l}\!)
 \!\! \left[\!a_{m} - \frac{\sigma^2 \zeta_l(b^{(n)}_m + v^{(n+1)}_m) + \eta_l s_m}{\sigma^2 \zeta_l + \eta_l} \right]^2 \nonumber\\
\! & \left. +  \!(\frac{1}{2 \eta_l} \!+\! \frac{1}{2 \sigma^2 \zeta_l}\!)
 \!\! \left [\!a_{m+M} \! - \! \frac{\sigma^2 \zeta_l(b^{(n)}_{m+M}\!+\!v^{(n+1)}_{m+M}) \!+\! \eta_l s_{m+M}}{\sigma^2 \zeta_l + \eta_l} \! \right]^2 \!\! \right\} \!.  \label{longg}
\end{align}
Define $\tilde{\mathbf{a}}_m = \left[a_{m},a_{m+M}\right]^\top$ and $g(\tilde{\mathbf{a}}_m) = \mathbb{I}(\tilde{\mathbf{a}}_m) + \lambda_l \| \tilde{\mathbf{a}}_m \|_2$.  
Note that (\ref{longg}) can be interpreted as finding the proximal mapping of $ \frac{2 \eta_l \sigma^2 \zeta_l}{\eta_l + \sigma^2 \zeta_l}g(\cdot)$. 
Define $\mathbf{c}^{(n)} = \left[ \frac{\sigma^2 \zeta_l(b^{(n)}_m + v^{(n+1)}_m) + \eta_l s_m}{\sigma^2 \zeta_l + \eta_l} , \frac{\sigma^2 \zeta_l(b^{(n)}_{m+M}+v^{(n+1)}_{m+M}) + \eta_l s_{m+M}}{\sigma^2 \zeta_l + \eta_l}\right]^\top$. 
The solution of (\ref{longg}) is then given by 
\begin{align}
 a^{(n+1)}_m \! &= \!\max \!\left \{0, \left[1 \!- \!\frac{2\eta_l\sigma^2 \zeta_l\lambda_l}{ (\sigma^2 \zeta_l + \eta_l) \max\{\|\mathbf{c}^{(n)}\|_2, \frac{2\eta_l\sigma^2 \zeta_l\lambda_l}{\eta_l+\sigma^2 \zeta_l}\} } \! \right]  \right. \nonumber\\
 &\left. \times \frac{\sigma^2 \zeta_l(b^{(n)}_m + v^{(n+1)}_m) + \eta_l s_m}{\sigma^2 \zeta_l + \eta_l} \right \}, \label{final00} \\ 
  a^{(n+1)}_{m+M} \! &= \! \max \!\left \{0, \left[1 \!-\!\frac{2\eta_l\sigma^2 \zeta_l\lambda_l}{ (\sigma^2 \zeta_l + \eta_l) \max\{\|\mathbf{c}^{(n)}\|_2, \frac{2\eta_l\sigma^2 \zeta_l\lambda_l}{\eta_l+\sigma^2 \zeta_l}\} } \! \right]  \right. \nonumber\\
 &\left. \times \frac{\sigma^2 \zeta_l(b^{(n)}_{m+M} + v^{(n+1)}_{m+M}) + \eta_l s_{m+M}}{\sigma^2 \zeta_l + \eta_l} \right \}. \label{final}
\end{align} 
By solving (\ref{longg}) for $m = 1,\cdots,M$, we can obtain the solution to (\ref{aaaaa}) whose elements are listed in  (\ref{final00}) and (\ref{final}). 
Furthermore, by repeating (\ref{viter})-(\ref{bbbbb}) until the convergence of ADMM, the solution to (\ref{prox}) is given by the final $\mathbf{v}^{(n)}$ in (\ref{viter}).  
% after the convergence of ADMM

\textcolor{blue}{%
\begin{remark}
The computational complexity of calculating $ F_l(\cdot) $ in (\ref{prox}) is dominated by the update of $\mathbf{v}^{(n)}$ in ADMM.
Suppose the maximum iteration number is $N_{ADMM}$. Then the complexity of calculating $ F_l(\cdot) $ is $ O (N_{ADMM} (U I N_r K M ^2 + M ^3) )$ in the worst case.
However, if we solve (\ref{prox}) through
CVX using the default interior point method, then the order of computational complexity is 
$O ( M^{3.5} \log(1/\rho))$ according to \cite{cvx}, where $\rho$ denotes the target tolerance. 
Since $M$ is typically much larger than $U, I, N_r$, and $K$, the proposed ADMM approach enjoys lower computational complexity than that with CVX. 
\end{remark}
}%

\subsection{Finding MACE by Solving Fixed-Point Problem}
Since the MACE equations (\ref{def0}) and (\ref{sum0}) are difficult to solve straightforwardly, we will further elucidate the solution process.

Note that even with the same input, the output of $ F_l(\cdot) $ may be different across BSs due to the difference in $\tilde{\mathbf{E}}_l$. 
Thus, let $\mathbf{s}_l$ denote the RPCD vector reconstructed by the $l$-th BS. 
We horizontally stack the vectors $\mathbf{s}_l \, ,l=1,\cdots,L,$ to form the overall MACE state as
\begin{align} 
\mathbf{S} = [\mathbf{s}_1, \cdots, \mathbf{s}_L].
\end{align}
Define the unified proximal operator by horizontally stacking the proximal mappings $F_l(\cdot) \, ,l=1,\cdots,L,$ as 
\begin{align} 
\mathbf{F}(\mathbf{S}) = [F_1(\mathbf{s}_1), \cdots, F_L(\mathbf{s}_L)].
\label{defF}
\end{align}

Since each column of $\mathbf{F}(\mathbf{S})$ is individually calculated by a certain BS, $\mathbf{F}(\mathbf{S})$ contains multiple and possibly inconsistent reconstructions. 
In order to yield a single RPCD reconstruction, we define the averaging operator as 
\begin{align} 
\mathbf{G}(\mathbf{S})=\frac{1}{L}\mathbf{S} \mathbf{1}_{L} \mathbf{1}_{L} ^\top = \left[\bar{x}(\mathbf{S}), \cdots, \bar{x}(\mathbf{S}) \right],
\end{align}
where $\bar{x}(\mathbf{S})=\frac{1}{L} \mathbf{S} \mathbf{1}_{L} 
=\frac{1}{L} \sum_{l=1}^{L} \mathbf{s}_l$.
The operator $\mathbf{G}(\mathbf{S})$ computes the average of the components in $\mathbf{S}$, and then returns the matrix formed by duplicating $\bar{x}(\mathbf{S})$ $L$ times. Note that the averaging operator $\mathbf{G}(\mathbf{S})$ is idempotent, i.e., $\mathbf{G}(\mathbf{G}(\mathbf{S}))=\mathbf{G}(\mathbf{S})$. 
Since $\mathbf{s}_l = \mathbf{s}^* + \mathbf{u}_l^*$ when MACE is achieved,
the MACE equations, i.e., (\ref{def0}) and (\ref{sum0}), can be reformulated compactly as 
\begin{align} 
\mathbf{F}\left(\mathbf{S}^*\right)=\mathbf{G}\left(\mathbf{S}^*\right),
\label{mace}
\end{align}
where $\mathbf{S}^*$ is the solution of (\ref{mace}).
Equation (\ref{mace}) can be interpreted as finding a consensus equilibrium among all BSs. 
The final MACE reconstruction is then given by $\mathbf{s}^*=\bar{x}\left(\mathbf{S}^*\right)$, and the columns of  $\mathbf{F}\left(\mathbf{S}^*\right)-\mathbf{S}^*$ are $\mathbf{u}_l^*$ in (\ref{sum0}), $l = 1,\cdots,L,$ whose summation is $\mathbf{0}_{2M}$. 

% The average of each BS is the consensus  solution  $\bar{x}\left(\mathbf{S}^*\right)$, 

The MACE equation (\ref{mace}) can be solved in a convenient way by converting the equation 
to a fixed-point problem, which can be efficiently and effectively solved by using the Mann iteration method \cite{mann}.
Define the auxiliary variable $\mathbf{Q}^*$ as 
\begin{align}
\mathbf{Q}^*=(2 \mathbf{G}-\mathbf{I}) \mathbf{S}^*. 
\label{qd}
\end{align}
Since $\mathbf{G}$ is idempotent, there is $(2 \mathbf{G}-\mathbf{I})^2 = \mathbf{I}$, i.e., $(2 \mathbf{G}-\mathbf{I})$ is self-inverse. 
From (\ref{qd}), we have $\mathbf{S}^*=(2 \mathbf{G}-\mathbf{I}) \mathbf{Q}^*$. Moreover, with (\ref{mace}), there is
\begin{align}
(2 \mathbf{F}-\mathbf{I}) \mathbf{S}^* = (2 \mathbf{G}-\mathbf{I}) \mathbf{S}^*. 
\label{qd2}
\end{align}
Substituting $\mathbf{S}^*=(2 \mathbf{G}-\mathbf{I}) \mathbf{Q}^*$ into (\ref{qd2}) and using the self-inverse property of $2 \mathbf{G}-\mathbf{I}$, we obtain 
\begin{align}
(2\mathbf{F} - \mathbf{I})(2\mathbf{G} - \mathbf{I}) \mathbf{Q}^* = \mathbf{Q}^*. 
\label{qd3}
\end{align}
Therefore, the solution to the MACE equation (\ref{mace}) is exactly the fixed point of the composition operator $\mathbf{T} = (2\mathbf{F} -
\mathbf{I})(2\mathbf{G} - \mathbf{I})$ given by 
\begin{align}
\mathbf{T} (\mathbf{Q}^*)=\mathbf{Q}^* . 
\end{align}
We can calculate the fixed-point $\mathbf{Q}^*$ using the Mann iterations \cite{mann}. 
In this approach, we start with an arbitrary initial guess $\mathbf{Q}^{(0)}$ and employ the following iterations:
\begin{align}
\mathbf{Q}^{(k+1)}=\rho \mathbf{T} (\mathbf{Q}^{(k)})+(1-\rho) \mathbf{Q}^{(k)}, k \geq 0,
\label{mann}%
\end{align} 
which guarantees the convergence to the fixed-point $\mathbf{Q}^*$ under the condition that the relaxation parameter $\rho$ is within the interval $(0,1)$ and the operator $\mathbf{T}$ is non-expansive \cite{mann}. 

The non-expansive property of $\mathbf{T}$ is verified as follows. 
Since each $F_l(\cdot)$ is the proximal mapping of a convex function $h_l(\cdot)$,
$F_l(\cdot)$ belongs to the class of resolvents. 
Therefore, both the concatenated proximal mapping resolvent $\mathbf{F}$ and the reflection resolvent $2 \mathbf{F}-\mathbf{I}$ are non-expansive \cite{mann}. 
Besides, it can be easily proven that the operator $2 \mathbf{G}-\mathbf{I}$ is equivalent to right-multiplying $\mathbf{Q}^{(k)}$ by the matrix $(\frac{2}{L} \mathbf{1}_L \mathbf{1}_L^\top - \mathbf{I})$, which is the Householder transformation that belongs to orthogonal transformation and 
does not change the distance between two vectors.
Thus, $2 \mathbf{G}-\mathbf{I}$ is also non-expansive. 
Consequently, $\mathbf{T}$, as the composition operator of $2 \mathbf{G}-\mathbf{I}$ and $2 \mathbf{F}-\mathbf{I}$, is also non-expansive, and thus  iteration (\ref{mann}) is guaranteed to converge. 

\begin{algorithm}[t] 
\caption{Multi-BS Data Fusion}
\begin{algorithmic}[1]
\State \textbf{Input:} Initial guess $\mathbf{S}^{(0)}$, $\tilde{\mathbf{E}}_l^{(0)}$, convergence tolerance $\epsilon$, relaxation parameter $\rho \in (0,1)$ 
\State \textbf{Output:} Consensus equilibrium $\mathbf{s}^*$
\State Initialize iteration counter $k \leftarrow 0$
\State Compute $\mathbf{Q}^{(0)} \leftarrow (2\mathbf{G} - \mathbf{I})\mathbf{S}^{(0)}$
\Repeat
    \State Compute $\mathbf{T}^{(k)}(\mathbf{Q}^{(k)}) \leftarrow (2\mathbf{F}^{(k)} - \mathbf{I})(2\mathbf{G} - \mathbf{I})\mathbf{Q}^{(k)}$
    \State Compute $\mathbf{Q}^{(k+1)} \leftarrow \rho\mathbf{T}^{(k)}(\mathbf{Q}^{(k)}) + (1-\rho)\mathbf{Q}^{(k)}$
    \State Compute $\mathbf{S}^{(k+1)} \leftarrow (2 \mathbf{G} - \mathbf{I}) \mathbf{Q}^{(k+1)}$
    \State Update $\tilde{\mathbf{E}}_l^{(k+1)}$ with $\mathbf{S}^{(k+1)}$ based on (\ref{sense1}), (\ref{ee}), and (\ref{assemble})
    \State Update $F_l^{(k+1)}(\cdot)$ with $\tilde{\mathbf{E}}_l^{(k+1)}$ based on (\ref{prox}) 
    \State Update $\mathbf{F}^{(k+1)}(\cdot)$ with $F_l^{(k+1)}(\cdot)$ based on (\ref{defF})
    \State $k \leftarrow k + 1$
\Until{$\|\mathbf{Q}^{(k+1)} - \mathbf{Q}^{(k)}\|_F < \epsilon$}
\State Compute $\mathbf{S}^* \leftarrow (2\mathbf{G} - \mathbf{I}) \mathbf{Q}^{(k)}$ 
\State Compute final EM property reconstruction $\mathbf{s}^* \leftarrow \bar{x}(\mathbf{S}^*)$
\end{algorithmic}
\label{alg}
\end{algorithm}

\textcolor{blue}{%
Since $\tilde{\mathbf{E}}_l$ and $F_l(\cdot)$ are related to $\mathbf{s}$, we integrate the Born iterative method (BIM) \cite{lipp,lipp2,operator,m-born,v-born} into MACE.
During the iterative process of EM property sensing based on BIM, updating $\tilde{\mathbf{E}}_l$ and the function $F_l(\cdot)$ is crucial to refine EM property estimates.
Initially, $\tilde{\mathbf{E}}_l$ is estimated using the Born approximation \cite{born,born2}, 
which assumes that the scattering electric field is relatively minor compared to the incident electric field. 
With the Born approximation, there is}%
\textcolor{blue}{%
\begin{align}
\mathbf{D}_{k,l}\!\! &= \!\!\left\{\bar{\mathbf{H}}_{1,k}^\top 
\!\! \left[\mathbf{I} + \sum_{i=1}^{\infty} \left( (\mathbf{G}_k \textrm {diag}\left ({\boldsymbol{\chi }_{k}})^{-\top}\right)^{i} \right)\right]\right\} \circ (\mathbf{P}_{k,l} \mathbf{H}_{2,k,l}) \nonumber \\
&\approx \bar{\mathbf{H}}_{1,k}^\top \circ (\mathbf{P}_{k,l} \mathbf{H}_{2,k,l}) .   
\label{sense2}%
\end{align}}%
\color{blue}
The Born approximation (\ref{sense2}) allows for the calculation of the initial sensing matrix $\tilde{\mathbf{E}}_l^{(0)}$, which is then used to reconstruct the initial guess $\mathbf{S}^{(0)}$.
In each BIM iteration, $\tilde{\mathbf{E}}_l^{(k)}$ is updated to reflect the latest estimates of EM property, and $F_l^{(k)}(\cdot)$ is adjusted accordingly to compensate for the errors introduced by approximation (\ref{sense2}). 
This BIM adjustment leverages recursive linear approximations, which makes the RPCD reconstruction converge to accurate EM property of the target \cite{lipp,lipp2,operator,m-born,v-born}. 
\color{black}

The proposed iteration procedure of multi-BS data fusion based on BIM is summarized in Algorithm \ref{alg}. 
\textcolor{blue}{%
Note that only $\mathbf{S}^{(k)}$ is transmitted between BSs and the CPU while $\tilde{\mathbf{E}}_l^{(k)}$ and $F_l^{(k)}(\cdot)$ can be updated by each BS locally without being transmitted to the CPU.
The BSs only need to transmit a $ 2 M $-dimensional vector $\mathbf{s}_l^{(k)}$ to the CPU. 
If it takes $N_{fusion}$ iterations for Algorithm \ref{alg} to converge, then the transmission overhead is $2 M N_{fusion}$ real-valued numbers.
Comparatively, if the BSs adopt signal-level fusion that requires transmitting the general sensing matrix $\tilde{\mathbf{E}}_l$, a total of $4 U I N_r K M$ complex-valued numbers need to be transmitted. 
Since $2 M N_{fusion} < 4 U I N_r K M $ in most cases, the transmission overhead is decreased in Algorithm \ref{alg} compared to signal-level data fusion. 
}%
% It manages computational complexity through predominantly linear operations, with the exception of the inversion step that is crucial for the algorithm's progression. 
% In the case that each $F_l(\cdot)$ is the proximal mapping for a convex, differentiable ${ }^1$ function $f_j$, then $x^*$ satisfies
% \begin{align}
% \sum_{l=1}^{L}
% \nabla h_l\left(\mathbf{S}^*\right) = 0,
% \end{align}
% and $\mathbf{S}^*$ is a minimizer of the function $f=\sum_j h_l$ [27]. 

% The computational complexity of Algorithm \ref{alg} is analyzed as follows
% \begin{itemize}
%     \item \textbf{The Cost of Applying $\mathbf{T}$:}
%     The operator $\mathbf{T}$ includes the application of both $\mathbf{F}$ and $\mathbf{G}$. Thus, the computational complexity of applying $\mathbf{T}$ once is the sum of the complexities of applying $\mathbf{F}$ and $\mathbf{G}$, that is $ O(F) + O (G)$.
        
%     \item \textbf{Arithmetic Operations:} The algorithm performs arithmetic operations (linear combinations) which are typically linear with respect to the number of elements in $\mathbf{Q}$, denoted as $N$. Thus, the complexity for these operations is $O(N)$.
    
%     \item \textbf{Iteration Count:} Let $N_{iter}$ denote the total number of iterations needed for the algorithm to converge. Assuming the cost per iteration remains constant, the overall complexity depends linearly on $K$.
% \end{itemize}

The complexity of Algorithm \ref{alg} is dominated by calculating $\mathbf{T}^{(k)}(\mathbf{Q}^{(k)})$ and updating $\mathbf{F}^{(k)}(\cdot)$. 
The complexity of calculating $\mathbf{T}^{(k)}(\mathbf{Q}^{(k)})$ is then dominated by calculating $\mathbf{F}^{(k)}(\mathbf{Q}^{(k)})$, i.e., $ O ( L N_{ADMM} (U I N_r K M ^2 + M ^3) )$.
Since the complexity of constructing  $\mathbf{D}_{k,l}$ is $ O( M ^ 3 + M ^ 2 U I + U I N_r M )$,
the complexity of updating $\tilde{\mathbf{E}}_l^{(k)}$
is $ O( K( M ^ 3 + M ^ 2 U I + U I N_r M ) )$. 
The complexity of updating $\mathbf{F}^{(k+1)}(\cdot)$ is then 
$ O( K L( M ^ 3 + M ^ 2 U I + U I N_r M ) )$.
Therefore, the general computational complexity of Algorithm \ref{alg} with $ N_{iter}$ iterations 
can be summarized as 
$ O( N_{iter} [ L N_{ADMM} (U I N_r K M ^2 + M ^3) + K L( M ^ 3 + M ^ 2 U I + U I N_r M ) ] ) $.

\section{Pilot Design and Material Identification} 

\subsection{Joint Design of Multi-Subcarrier Pilots}
% In order to increase the diversity of the pilot patterns in region $D$, we focus on 
% minimizing the mutual coherence of the equivalent channels from the $u$-th UE to the target $\tilde{\mathbf{H}}_{1,k,u} = \mathbf{H}_{1,k,u} \tilde{\mathbf{W}}_{k,u}$ between multiple subcarriers and multiple measurements. 
% Mutual coherence measures the correlation or linear dependence between signals; reducing this helps in improving the clarity and distinctiveness of each pilot.       
% Since the UEs are typically unaware of others' CSI, each UE designs the pilots independently. 

% This is applied across multiple subcarriers and multiple measurements. 
\textcolor{blue}{%
Since there is no closed-form expression for the mean squared error (MSE) of RPCD reconstruction, we use mutual coherence as an alternative pilot design criterion.
In the context of compressive sensing, mutual coherence, which measures the correlation or linear dependence between signals or patterns, is a crucial metric used to evaluate the accuracy of RPCD reconstruction \cite{adaptive,adaptive2}.
By minimizing the column mutual coherence of the equivalent channels from the \(u\)-th UE to the target, i.e., \(\tilde{\mathbf{H}}_{1,k,u} \overset{\Delta}{=} \mathbf{H}_{1,k,u} \tilde{\mathbf{W}}_{k,u}\), we can enhance the diversity of incident EM wave patterns within region \(D\). 
With lower mutual coherence and more diverse EM wave patterns, richer information of the target's EM property can be extracted, which is crucial for effective EM property sensing.
}%
Since the UE typically does not have access to the other UEs' CSI, each UE should independently design its pilot patterns.

Define $\mathbf{A}_{k,u}=\left[\mathbf{H}_{1,1,u} \tilde{\mathbf{W}}_{1,u}, \cdots, \mathbf{H}_{1,k-1,u} \tilde{\mathbf{W}}_{k-1,u} \right] $, and define $\mathbf{A}_{1,u} = \mathbf{0}_{M}$ as a special case for $k = 1$.
\color{blue}
For each UE, we sequentially design the pilots for the $k$-th subcarrier from $1$ to $K$
based on the following optimization:
\color{black}
\begin{subequations}
\begin{align}
\min_{\tilde{\mathbf{W}}_{k,u},\xi_{k,u}}  \!\!&& \!\!\!\! \!\!\!\!
p(\tilde{\mathbf{W}}_{k,u},\xi_{k,u}) 
&= 
\left\|\tilde{\mathbf{W}}_{k,u}^H \mathbf{H}_{1,k,u}^H \mathbf{H}_{1,k,u} \tilde{\mathbf{W}}_{k,u}- \xi_{k,u} \mathbf{I} \right\|_F^2 \nonumber\\ 
&& &+ \gamma_{k,u} \left \| \mathbf{A}_{k,u}^H \mathbf{H}_{1,k,u} \tilde{\mathbf{W}}_{k,u} \right \|_F^2  \label{pgd0} \\
 \text { s.t. }  && \!\!\!\! \!\!\!\!
 \left \| \tilde{\mathbf{W}}_{k,u} \right \|_F^2 & \leq P_{k,u}, \label{pgd1} \\
  && I \xi_{k,u}  & \geq  \bar{P}_{k,u}   ,
\label{pgd1.5}%,  \xi_{k,u} \geq 0
\end{align}%
\label{pppp}%
\end{subequations} 
where $\gamma_{k,u}$ is a penalty factor, $\xi_{k,u}$ is a scaling factor that represents the overall power of the objective pilot patterns in region $D$, 
$P_{k,u}$ is the general power budget of the $u$-th UE for the $k$-th subcarrier, and $\bar{P}_{k,u}$ is the minimum overall power in region $D$. 

\color{blue}
In the objective function (\ref{pgd0}), the first term aims at the reduction of the mutual coherence among different symbols transmitted within the same subcarrier.
This is achieved by making the Gram matrix of the equivalent channel as orthogonal as possible, for each UE and on each subcarrier.
Therefore, the ideal outcome is that the Gram matrix resembles $\xi_{k,u} \mathbf{I}$. 
The second term in (\ref{pgd0}) aims at the reduction of the mutual coherence between symbols transmitted across different subcarriers.
By minimizing this norm, the optimization ensures that the pilot patterns in region \(D\) across different subcarriers also remain as orthogonal as possible to each other, which is crucial to improve the overall system performance in multi-subcarrier ISAC environments.
\color{black} 
Therefore, a fully digital beamforming architecture is necessitated to adjust $\tilde{\mathbf{W}}_{k,u}$ on different carriers. 

 % involving complex-valued variables
In order to solve the optimization problem (\ref{pppp}), we can apply the projected gradient descent (PGD) method combined with Wirtinger calculus to the objective function (\ref{pgd0}) \cite{Wirtinger}.
We first need to calculate the gradient of \( p(\tilde{\mathbf{W}}_{k,u},\xi_{k,u}) \) with respect to the conjugate of the variable \( \tilde{\mathbf{W}}_{k,u}^* \). Using Wirtinger derivatives, we obtain 
\begin{align}%
& \frac{\partial p(\tilde{\mathbf{W}}_{k,u},\xi_{k,u})}{\partial \tilde{\mathbf{W}}_{k,u}^*}  \!  =  \!2 \mathbf{H}_{1,k,u}^H \mathbf{H}_{1,k,u} \! \tilde{\mathbf{W}}_{k,u} \! \left(\tilde{\mathbf{W}}_{k,u}^H \mathbf{H}_{1,k,u}^H  \right. \! \nonumber\\
& \times \left. \mathbf{H}_{1,k,u} \! \tilde{\mathbf{W}}_{k,u} \!- \!  \xi_{k,u} \!\mathbf{I} \right) 
 \!+\! \gamma_{k,u} \mathbf{H}_{1,k,u}^H \mathbf{A}_{k,u} \mathbf{A}_{k,u}^H \mathbf{H}_{1,k,u} \!\tilde{\mathbf{W}}_{k,u}.
\label{pgd2}%
\end{align}

The PGD method involves two key steps: the gradient descent update and the projection onto the feasible set given by the constraints.
Specifically, the update process for $\mathbf{W}_{k,u}^{(t)}$ can be formulated as 
\begin{align}
\dot{\mathbf{W}}_{k,u}^{(t+1)} &= \tilde{\mathbf{W}}_{k,u}^{(t)} - \alpha_{k,u} \frac{\partial p(\tilde{\mathbf{W}}_{k,u}^{(t)},\xi_{k,u}^{(t)} ) }{\partial \tilde{\mathbf{W}}_{k,u}^*} , \label{pgd2.5}  \\
\tilde{\mathbf{W}}_{k,u}^{(t+1)} &= \min \left\{ 1 , \sqrt{\frac{P_{k,u}}{\left\|\dot{\mathbf{W}}_{k,u}^{(t+1)}\right\|_F^2}} \right\} \dot{\mathbf{W}}_{k,u}^{(t+1)},
\label{pgd3} 
\end{align}
where $\alpha_{k,u}$ is the step size determined by the Armijo backtracking line search method, 
and the superscript $(t)$ denotes the iteration number.

Next, the scaling factor is updated by directly minimizing (\ref{pgd0}) and projecting $\xi_{k,u}$ onto (\ref{pgd1.5}) as 
\begin{align}
\xi_{k,u}^{(t+1)} &= \max \left \{ \frac{1}{I} \bar{P}_{k,u},
\frac{1}{I} \text{Tr} \left [ 
 (\tilde{\mathbf{W}}_{k,u}^{(t)} ) ^H \mathbf{H}_{1,k,u}^H \mathbf{H}_{1,k,u} \tilde{\mathbf{W}}_{k,u}^{(t)} \right] \right \} \nonumber \\
 & = \frac{1}{I} \max \left \{ \bar{P}_{k,u} , \left \| \mathbf{H}_{1,k,u} \tilde{\mathbf{W}}_{k,u}^{(t)} \right\|_F ^2  \right \}.
\label{pgd4} 
\end{align}
The formulations (\ref{pgd2.5})-(\ref{pgd4}) ensure that each iterative update not only minimizes $p(\tilde{\mathbf{W}}_{k,u},\xi_{k,u})$ but also adheres to the constraints imposed by the power budget and the minimum power of the objective pilot patterns in region $D$.

The computational complexity of the joint design of multi-subcarrier pilots is dominated by the calculation of $\frac{\partial p(\tilde{\mathbf{W}}_{k,u},\xi_{k,u})}{\partial \tilde{\mathbf{W}}_{k,u}^*}$. 
Suppose the PGD method takes $N_{k,u}$ iterations to converge, and then the complexity of designing $\tilde{\mathbf{W}}_{k,u}$ is $ O ( N_{k,u} (  I N_t M + M^2 I+ I^2 N_t + I N_t^2 ) )$. 

\color{blue}
\begin{theorem}%[Convergence of PGD for Problem (\ref{pppp})]
The PGD algorithm applied to problem (\ref{pppp}) is guaranteed to converge to a stationary point, which can be either a local minimum or a saddle point.
\label{pgdt1}
\end{theorem}

\begin{proof}
The proof consists of establishing 4 key properties and then applying the convergence theorem for PGD algorithm on problem (\ref{pppp}).

1. Continuous differentiability of $p(\tilde{\mathbf{W}}_{k,u},\xi_{k,u})$:
The objective function is a multivariate polynomial and is thus continuously differentiable.

2. Lipschitz continuity of the gradient of $ p(\tilde{\mathbf{W}}_{k,u},\xi_{k,u})$:
The gradient of $p(\tilde{\mathbf{W}}_{k,u},\xi_{k,u})$ is also a multivariate polynomial.
Since $\mathbf{W}_{k,u}$ is finite due to the constraints and $\xi_{k,u}$ is also finite according to (\ref{pgd4}), the gradient is Lipschitz continuous on the constraint set.

3. Closed and convex constraint set:
\begin{enumerate}
% \begin{itemize}
    \item The constraint $\left \| \tilde{\mathbf{W}}_{k,u} \right \|_F^2 \leq P_{k,u}$ defines a closed and convex set (a ball in the Frobenius norm).
    \item The constraint $I \xi_{k,u} \geq \bar{P}_{k,u}$ defines a closed halfspace, which is convex.
    % \item The intersection of $\left \| \tilde{\mathbf{W}}_{k,u} \right \|_F^2 \leq P_{k,u}$ and $I \xi_{k,u} \geq \bar{P}_{k,u}$ is closed and convex.
% \end{itemize}
\end{enumerate}

4. Lower boundedness of (\ref{pgd0}) on the constraint set:
The objective function (\ref{pgd0}) consists of a sum of squared norms and is bounded from below by zero.

Given these 4 properties, we can apply the convergence theorem for PGD on closed and convex constraint set \cite{chen2019non}:

% Let $\{\tilde{\mathbf{W}}_{k,u}^{(t)},\xi_{k,u}^{(t)} \}, t = 1,2,\cdots$ be the sequence generated by PGD with a fixed
Suppose the step size $\alpha_{k,u} $ satisfies $0 < \alpha_{k,u} < \frac{1}{L'}$, where $L'$ is the Lipschitz constant of $\frac{\partial p(\tilde{\mathbf{W}}_{k,u}^{(t)},\xi_{k,u}^{(t)} ) }{\partial \tilde{\mathbf{W}}_{k,u}^*}$. Then:
\begin{enumerate}
    \item The sequence $\{p(\tilde{\mathbf{W}}_{k,u}^{(t)},\xi_{k,u}^{(t)})\}$ converges when $t \rightarrow \infty$.
    \item The sequence $\{\tilde{\mathbf{W}}_{k,u}^{(t)},\xi_{k,u}^{(t)}\}$ converges when $t \rightarrow \infty$.
    \item The limit point of $\{\tilde{\mathbf{W}}_{k,u}^{(t)},\xi_{k,u}^{(t)}\}$  is a stationary point of (\ref{pppp}) when $t \rightarrow \infty$.
\end{enumerate}

Therefore, the PGD algorithm applied to problem (\ref{pppp}) will converge to a stationary point, which can be either a local minimum or a saddle point. 
\end{proof} 

\begin{remark}
While Theorem \ref{pgdt1} guarantees the convergence to a stationary point, it does not guarantee that this point is a global minimum due to the non-convexity of problem (\ref{pppp}). In practice, running the PGD algorithm multiple times with different random initializations is beneficial to increase the likelihood of finding a good local minimum.
\end{remark}
\color{black}

\begin{figure*}[t]
  \centering
\begin{minipage}[t]{0.33\linewidth}
\subfigure[UE-1, subcarrier-1]{
\includegraphics[width=6cm,height=5.5cm]{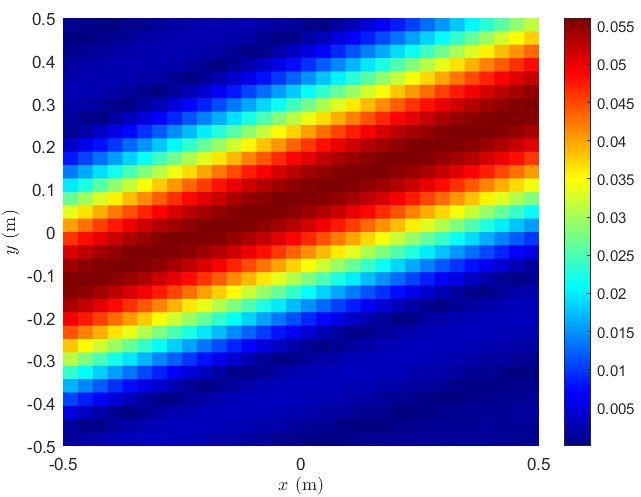}} %target1
\end{minipage}
\begin{minipage}[t]{0.33\linewidth}
\subfigure[UE-1, subcarrier-16]{
\includegraphics[width=6cm,height=5.5cm]{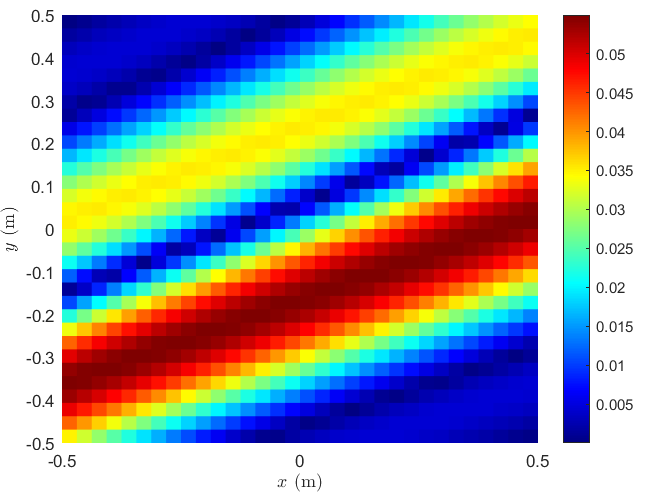}}  %image_halfmask.png
\end{minipage}  
\begin{minipage}[t]{0.32\linewidth}
\subfigure[UE-1, subcarrier-32]{
\includegraphics[width=6cm,height=5.5cm]{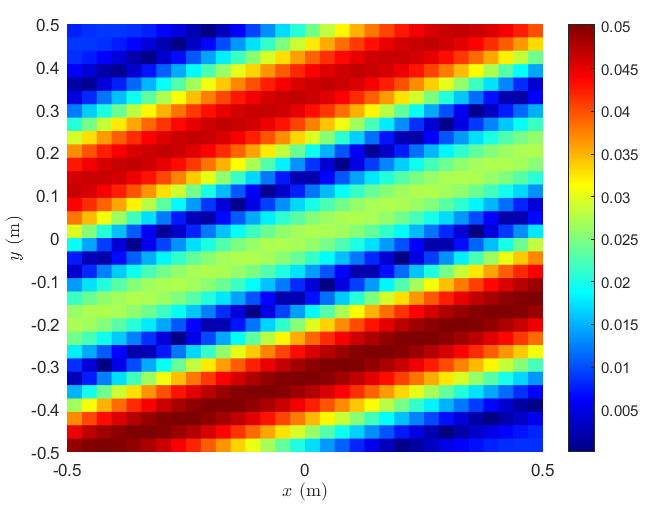}}   %image_fullmask.png
\end{minipage} \\       
\begin{minipage}[t]{0.33\linewidth}
\subfigure[UE-2, subcarrier-1]{
\includegraphics[width=6cm,height=5.5cm]{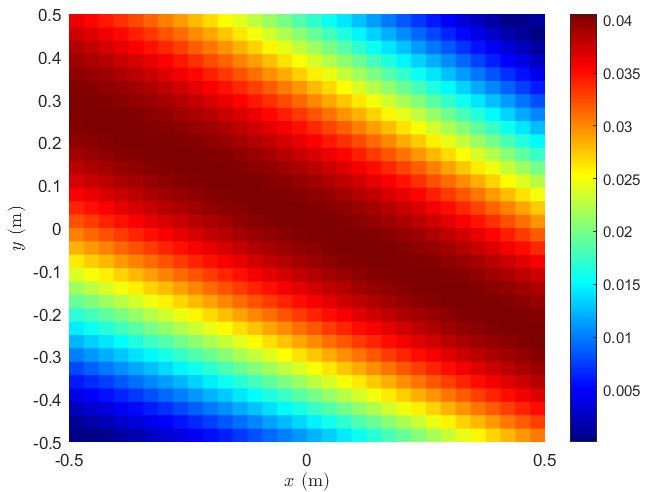}}
\end{minipage}
\begin{minipage}[t]{0.33\linewidth}
\subfigure[UE-2, subcarrier-16]{
\includegraphics[width=6cm,height=5.5cm]{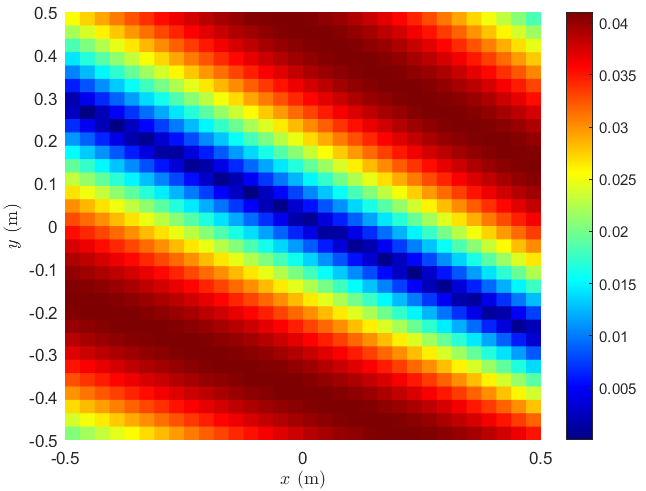}}
\end{minipage}
\begin{minipage}[t]{0.32\linewidth}
\subfigure[UE-2, subcarrier-32]{ 
\includegraphics[width=6cm,height=5.5cm]{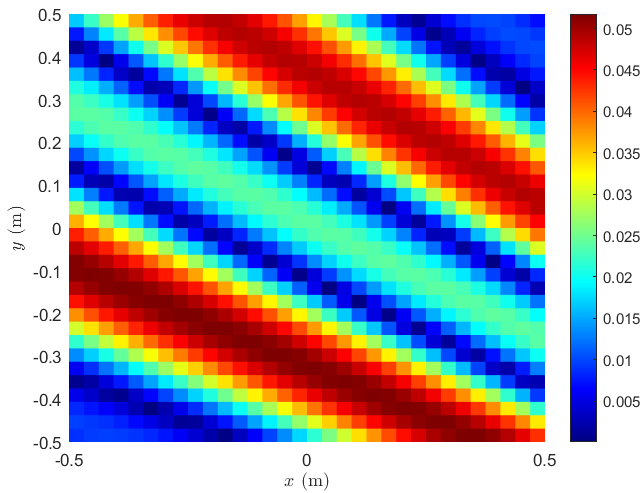}}
\end{minipage}
\caption{Pilot patterns in the region $D$ generated by different UEs and subcarriers demonstrated by $\left|\mathbf{E}_k^{i,D}\right|$. 
}
\label{pilot1}
\end{figure*}

% demonstrate
To provide a vivid illustration of the joint multi-subcarrier pilots design, we show in Fig.~\ref{pilot1} the pilot patterns $\left|\mathbf{E}_k^{i,D}\right|$ in region $D$. 
Suppose prior knowledge determines that the target is located within $ D = [-0.5,0.5] \times [-0.5,0.5]$~$\mathrm{m}^2$. 
We choose the number of sampling points as $M = 32 \times 32 = 1024$. 
There are $2$ UEs equipped with $8$-antennas uniform linear array (ULA). 
The inter-antenna spacing is set as $0.005$ m.
Two UEs are located at $(-9.3,-3.8)$ m and $(4.4,-2.4)$ m, respectively. 
The central frequency is $f_c = 30$ GHz and the frequency spacing of OFDM subcarriers is $\Delta f = 120$ KHz.
For each UE, pilots are transmitted in $K = 32$ subcarriers. 
For brevity, we only show the pattern of the first pilot symbol in each subcarrier. 

As shown in Fig. \ref{pilot1}, the different subcarriers provide distinct pilot patterns in region $D$.
However, due to the relatively small number of the UE antennas in practice, the stripes of the patterns generated by the same UE 
are roughly in the same direction that is mainly determined by the position of the UE. 
Thus, a single UE may not be sufficient to generate pilots for the reconstruction of RPCD within $D$. 
To address this problem, we require multiple users at different positions to provide diverse views and resolutions of RPCD. 
Besides, since UE-2 is closer to $D$, the near-field effect is more prominent and the pilot pattern stripes are not parallel.
Closer proximity leads to more diversity in the pattern style, which is more beneficial to RPCD reconstruction. 
Despite the fact that each UE designs its pilots independently in (\ref{pppp}), the pilot patterns of different UEs are inherently distinct and with low mutual coherence when the UEs are in different positions. 

\color{blue}
\begin{remark}
% OFDM contributes in reducing mutual coherence by enlarging the column space of $[\tilde{\mathbf{H}}_{1,1,u},\cdots , \tilde{\mathbf{H}}_{1,K,u}]$.
% Due to \(\tilde{\mathbf{H}}_{1,k,u} = \mathbf{H}_{1,k,u} \tilde{\mathbf{W}}_{k,u}\), 
% there is     
OFDM contributes in reducing mutual coherence by enlarging the column space of $[\tilde{\mathbf{H}}_{1,1,u}, \cdots , \tilde{\mathbf{H}}_{1,K,u}]$.
No matter how large $I$ is, the dimension of the column space of $\tilde{\mathbf{H}}_{1,k,u}$ for each $k$ is bounded by $\text{dim}(\text{col}(\tilde{\mathbf{H}}_{1,k,u})) = \text{dim}(\text{col}(\mathbf{H}_{1,k,u} \tilde{\mathbf{W}}_{k,u})) \leq N_t$ due to $\mathbf{H}_{1,k,u} \in \mathbb{C}^{M \times N_t}$, which is usually much smaller than $M$. 
However, the dimension of the column space of the combined matrix $[\tilde{\mathbf{H}}_{1,1,u},\cdots , \tilde{\mathbf{H}}_{1,K,u}]$ with $K$ subcarriers is bounded by $\text{dim}(\text{col}([\tilde{\mathbf{H}}_{1,1,u}, \cdots , \tilde{\mathbf{H}}_{1,K,u}])) \leq N_t K$.
The dimension of the column space increases significantly from not larger than $N_t$ to potentially $N_t K$, which effectively enlarges the space of pilot patterns generated by the $u$-th UE.
Consequently, the expansion of the column space reduces the likelihood of overlap between pilot patterns from the $u$-th UE, thereby decreasing mutual coherence in the OFDM sensing system.
\end{remark}
\color{black}

\subsection{Material Identification Methodology}
After RPCD reconstruction, the materials of the target can then be identified.
Suppose that the target is known to be constituted by  several possible materials, whose permittivity and conductivity are all precisely measured in advance.  
From the perspective of EM property sensing, 
only the materials with significant discrepancies in permittivity or conductivity can be distinguished. 
In detail, the material identification process consists of two steps: clustering and then classification. 

In order to determine the material of the target, we first need to distinguish between the parts of region $D$ occupied by different materials. 
To accomplish this, we utilize the density-based spatial clustering of applications with noise (DBSCAN) algorithm to divide the sampling points within $D$ into several categories \cite{dbscan1}. 
DBSCAN is an unsupervised algorithm with strong generalization of clusters for different shapes and sizes, and is guaranteed to converge regardless of the input data distribution.
The approach to determining the DBSCAN hyper-parameters has been discussed in \cite{dbscan3}.
Since the relative permittivity and conductivity have different dimensions, 
we adopt the non-dimensional and scale-invariant Mahalanobis distance as the metric in the DBSCAN clustering algorithm \cite{mahalanobis}.
The number of clusters is self-adaptive and represents the estimated number of material categories.  
The cluster centroid of the air, which represents the average permittivity and conductivity values of the air, is typically close to $(1,0)$. 
Conversely, the cluster centroids of the target materials, which individually represent their average permittivity and conductivity, are generally positioned significantly away from $(1,0)$.  

After DBSCAN clustering is performed, the next step is to determine the material categories of the target components. 
The procedure begins by eliminating outlier points from the data.
Subsequently, we calculate the Mahalanobis distances between the centroids of each cluster, which represent the target materials, and the measured EM property values of each candidate material. 
The target materials are subsequently assigned to the categories that correspond to the smallest Mahalanobis distance,
signifying the highest degree of similarity between the measured EM property of the target and the recognized characteristics of the known materials. 

% This is done by calculating the Mahalanobis distances between the cluster centroids of the target material and the ground-truth values of the permittivity and conductivity for each possible material.
% The target is then classified into the material categories with the shortest Mahalanobis distance which indicates the closest match between the measured EM property of the target cluster centroids and the known property of the materials. 

The computational complexity of the material identification is dominated by the DBSCAN clustering algorithm, whose complexity is $\mathcal{O}(M^2)$ in the worst case \cite{dbscan2} and can be further improved by applying acceleration techniques \cite{dbscan4}. 

\section{Simulation Results and Analysis}
% \begin{figure}[t]
%   \centering
% \centerline{\includegraphics[width=8.4cm,height=6.5cm]{mutual.png}}  
% \caption{\textcolor{black}{Mutual coherence of the general sensing matrix $\tilde{\mathbf{E}}$ with the increase of the number of subcarriers $K$.}}
% \label{mut}  
% \end{figure}

% In this section, we use the EM simulation software Ansys HFSS to simulate the transmission process in the ISAC system and to generate the received signals, which are used in reconstructing the EM property and identifying the material of the target. 
\textcolor{blue}{
Suppose the target is located within the region $D = [-1,1] \times [-1,1]$~$\mathrm{m}^2$ unless otherwise specified.
A total of $U = 10$ UEs are randomly located in a circular region of radius $10$~m centered at the origin but outside region $D$.
We consider a ULA transmitter equipped with $N_t = 8$ antennas at each UE. 
The power budget is the same for all UEs and is evenly allocated to all subcarriers.
For each UE, $ I = 16 $ pilot OFDM symbols are transmitted on each subcarrier. 
Assume $4$ BSs are located at $(100,0)$~m, $(0,100)$~m, $(-100,0)$~m, and $(0,-100)$~m for $l=1,2,3,4$, respectively, and are fused in sequence. 
We set $\zeta_l = 1/L, l = 1,\cdots, L.$
Each BS is equipped with $N_r = 64$ ULA antennas, and the ULA is perpendicular to the line from the origin to the BS center. 
We assume that the sum of thermal noise and possible interference from unknown EM sources or scatters 
at each BS is subject to the same zero-mean circularly symmetric complex Gaussian distribution with the same covariance.
SNR is then defined as the ratio of the overall received signal power from all BSs to the overall noise power. 
The central frequency is $f_c = 28$ GHz and the frequency spacing of OFDM subcarriers is $ \Delta f = 240$ KHz according to the 3GPP standard \cite{3gpp}. 
The inter-antenna spacing for both the UEs and the BSs is set as $\lambda_c/2 = 0.0054$ m.}
We choose the number of sampling points in region $D$ as $ M = 64 \times 64 = 4096$. 
Each sampling point within $D$ can be regarded as a single pixel in an image, i.e., the RPCD image consists of $64 \times 64 = 4096$ pixels.

Besides, we introduce the normalized mean square error (NMSE) of RPCD reconstruction as the criterion to evaluate the performance of EM property sensing 
\begin{align}
\mathrm{NMSE}&=
10 \log_{10} 
\frac{\left\|\mathbf{s}-\mathbf{s}^*\right\|_2^{2}}{\left\|\mathbf{s}\right\|_2^{2}}  \nonumber\\
&= 
10 \log_{10}
\frac{\left\|\boldsymbol{\epsilon}_r-\boldsymbol{\epsilon}_r^*\right\|_2^{2} + \frac{1}{\omega_c^2 \epsilon_0^2}\left\|\boldsymbol{\sigma}-\boldsymbol{\sigma}^*\right\|_2^{2}} {\left\|\boldsymbol{\epsilon}_r\right\|_2^{2} + \frac{1}{\omega_c^2 \epsilon_0^2}\left\|\boldsymbol{\sigma}\right\|_2^{2}}.
\label{NMSE}
\end{align} 
In (\ref{NMSE}), the NMSE of RPCD combines the estimation errors of both $\boldsymbol{\epsilon}_r$ and $\boldsymbol{\sigma}$,
which assesses the EM property sensing performance in a comprehensive way.  
Moreover, the material classification accuracy is defined as the ratio of the number of correctly classified samples to the total number of tested samples.

% We define SNR in the multi-BS system as the ratio of the sum of signal power from all BSs to the sum of noise power. 

% \begin{figure}[t]
%   \centering
% \centerline{\includegraphics[width=8.4cm,height=6.5cm]{EDOF.png}}
%   \caption{\textcolor{black}{Change of EDOF with the increase of the number of sampling points $M$ in region $D$.}}
%   \label{edof}
% \end{figure}

\begin{figure*}[t]
  \centering
\begin{minipage}[t]{0.33\linewidth}
\subfigure[Target real relative permittivity]{
\includegraphics[width=6cm,height=5.5cm]{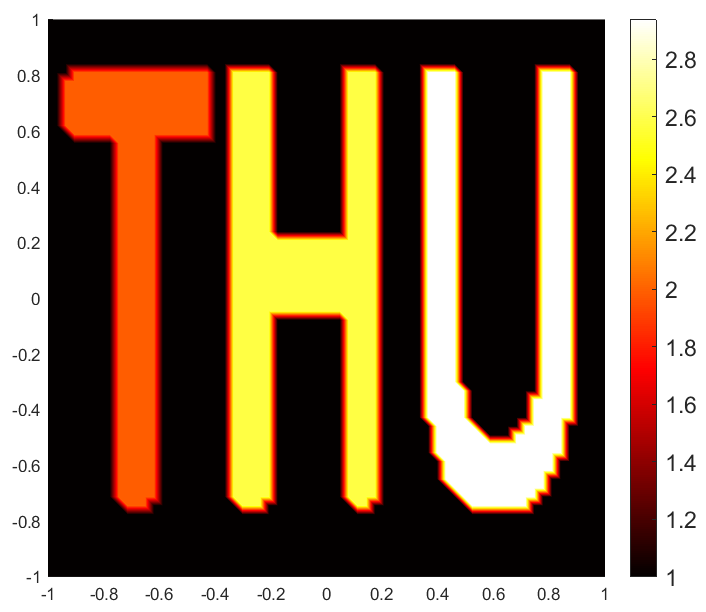}} %target1
\end{minipage}
\begin{minipage}[t]{0.33\linewidth}
\subfigure[Reconstructed relative permittivity \protect\\ with SNR = 0 dB]{
\includegraphics[width=6cm,height=5.5cm]{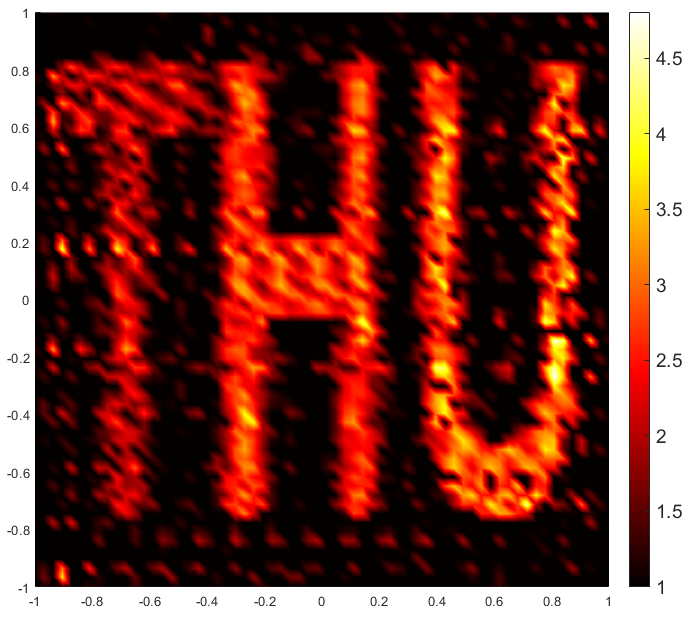}} 
\end{minipage}  
\begin{minipage}[t]{0.32\linewidth}
\subfigure[Reconstructed relative permittivity \protect\\ with SNR = 30 dB]{
\includegraphics[width=6cm,height=5.5cm]{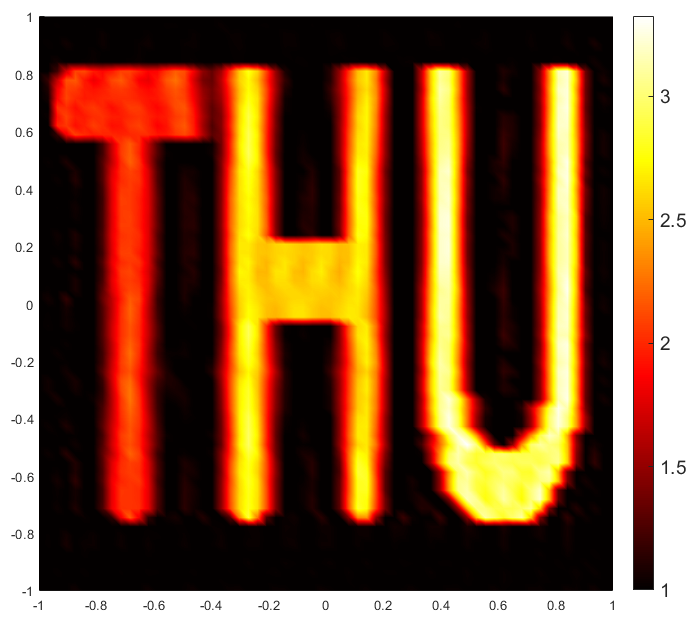}}   %image_fullmask.png
\end{minipage} \\              
\begin{minipage}[t]{0.33\linewidth}
\subfigure[Target real conductivity]{
\includegraphics[width=6cm,height=5.5cm]{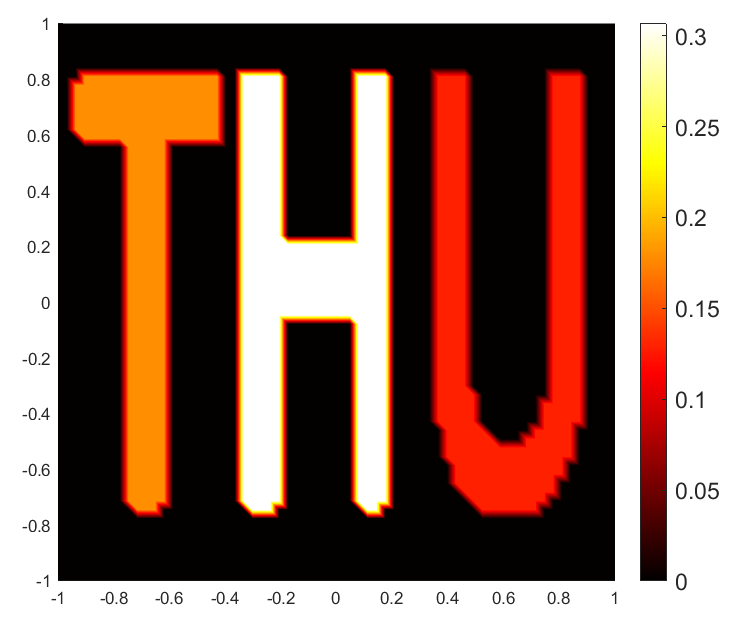}}
\end{minipage}
\begin{minipage}[t]{0.33\linewidth}
\subfigure[Reconstructed conductivity   with SNR = 0 dB]{
\includegraphics[width=6cm,height=5.5cm]{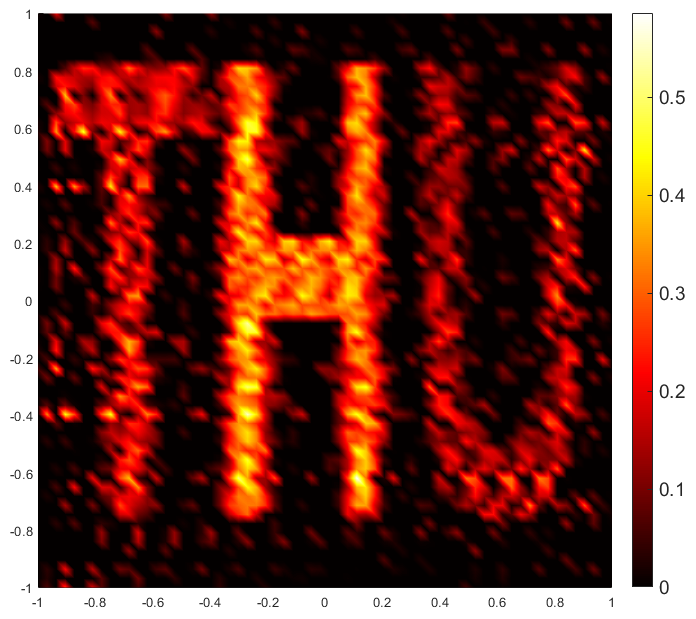}}
\end{minipage}
\begin{minipage}[t]{0.32\linewidth}
\subfigure[Reconstructed conductivity  with SNR = 30 dB]{
\includegraphics[width=6cm,height=5.5cm]{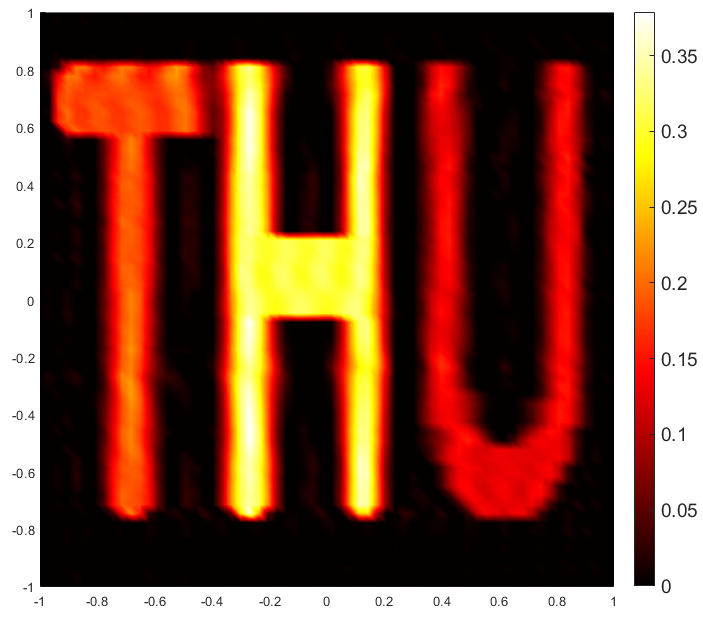}}
\end{minipage}
\caption{RPCD reconstruction results versus SNR with $L=4$ BSs and $ K= 64 $ subcarriers. 
Unit of conductivity is S/m.
}
\label{image1}
\end{figure*}

\begin{figure}[t]
  \centering
\centerline{\includegraphics[width=8.4cm,height=6.5cm]{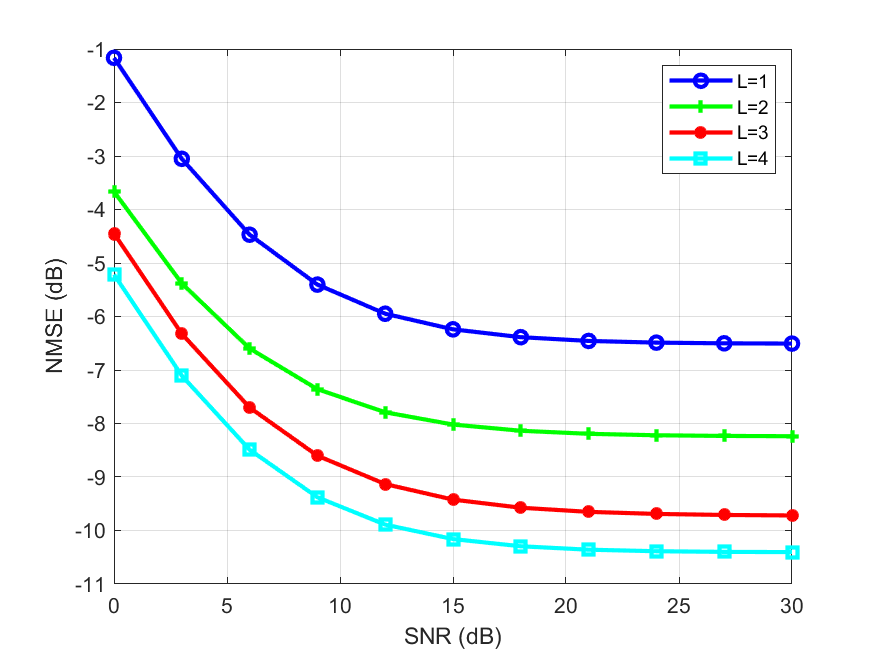}}
  \caption{\textcolor{black}{NMSE of PRCD versus SNR with $K = 32$.}}
  \label{nmse1}
\end{figure}

% Correspondingly, the change of mutual coherence of the general sensing matrix $\tilde{\mathbf{E}}$ with the increase of the number of subcarriers $K$ is demonstrated in Fig.~\ref{mut}. 
% It is seen that $\mu$ decreases with the increase of $K$, which indicates that the multi-frequency scheme can enhance the EM property sensing performance. 
% When $K$ reaches a certain threshold, $\mu$ scarcely changes, where the transmitter-target 
% distance $d$ is the decisive factor of sensing performance.
% When the distance becomes larger, the mutual coherence of $\tilde{\mathbf{E}}$ also becomes larger, which indicates worse sensing quality.

% The elevation angle and azimuth angle of the receiver with respect to the target are $\theta_{\text{r}} =60\degree$ and $\phi_{\text{r}} =0\degree$, respectively.

% \textcolor{black}{
% Since the frequency of EM waves is in the infrared band, the image of the target can not be recovered by conventional optical cameras. 
% }

\color{blue}
\subsection{RPCD Reconstruction Performance versus SNR}
To illustrate the RPCD reconstruction results vividly, we present the reconstructed images of the target 
based on relative permittivity and conductivity, respectively, with $L = 4$, $K = 32$, and SNR $=0$ dB or $30$ dB in Fig.~\ref{image1}.
The target ``THU" is the abbreviation of ``Tsinghua University", and the three letters are composed of wood, chipboard, and plasterboard, respectively.
It is seen from Fig.~\ref{image1} that, the reconstructed RPCD can reflect the general shape of the target. 
The RPCD reconstructed at SNR $ = 30$ dB is much more accurate to demonstrate the target's shape compared to the RPCD reconstructed at SNR $ = 0$ dB.
Moreover, a higher SNR value results in more accurate reconstructed values of relative permittivity and conductivity, which leads to better representation of the target's real EM property.  

We explore the NMSE of reconstructed RPCD versus the SNR, as shown in Fig.~\ref{nmse1}. 
We set the number of subcarriers to be $K = 32$ and sequentially select several BSs and fuse their data at the feature level to enhance the overall accuracy of the reconstructed RPCD. 
It is seen from Fig.~\ref{nmse1} that, the NMSE decreases with the increase of SNR in all cases, and the best performance is achieved when $L = 4$.  
The NMSE decreases fast in the beginning when the SNR increases from $0$~dB to $10$~dB.  
When the SNR reaches a threshold of approximately $20$ dB, the NMSE reduces to an error floor and remains largely unchanged.
At this level, the limitation on RPCD reconstruction quality is no longer due to noise at the BS but rather to the 
insufficient information about the target captured in the general sensing matrix.
As signals from more BSs are utilized, more information of the target is integrated, which leads to a smaller error floor in RPCD reconstruction. 
% Moreover, the error floor is reached at lower SNR levels. 

\begin{figure*}[t]
\captionsetup[subfigure]{singlelinecheck=false} 
  % \centering
\begin{minipage}[t]{0.33\linewidth}
\captionsetup[subfigure]{singlelinecheck=false}
% \centering 
\subfigure[Target relative permittivity]{
\includegraphics[width=6cm,height=5.5cm]{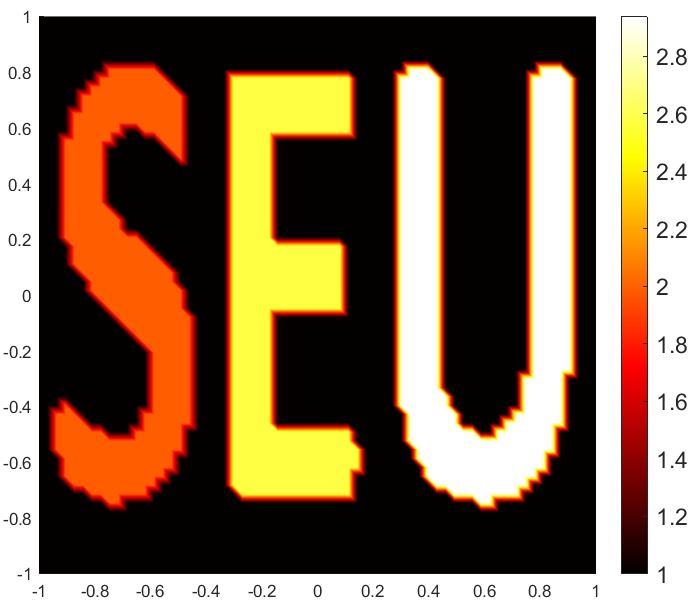}} % target1
\end{minipage}
\begin{minipage}[t]{0.33\linewidth}
\subfigure[Reconstructed relative permittivity with $K$ = 16]{
\includegraphics[width=6cm,height=5.5cm]{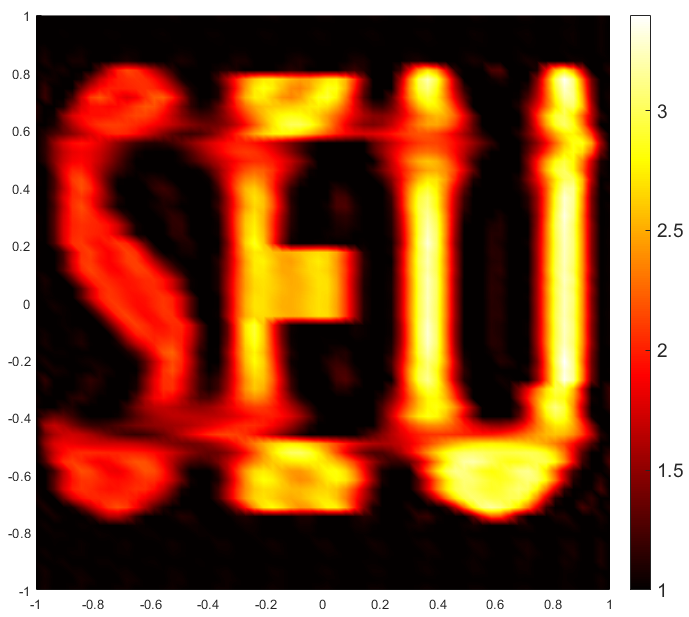}}  %image_halfmask.png
\end{minipage}  
\begin{minipage}[t]{0.32\linewidth}
\subfigure[Reconstructed relative permittivity  with $K$ = 64]{
\includegraphics[width=6cm,height=5.5cm]{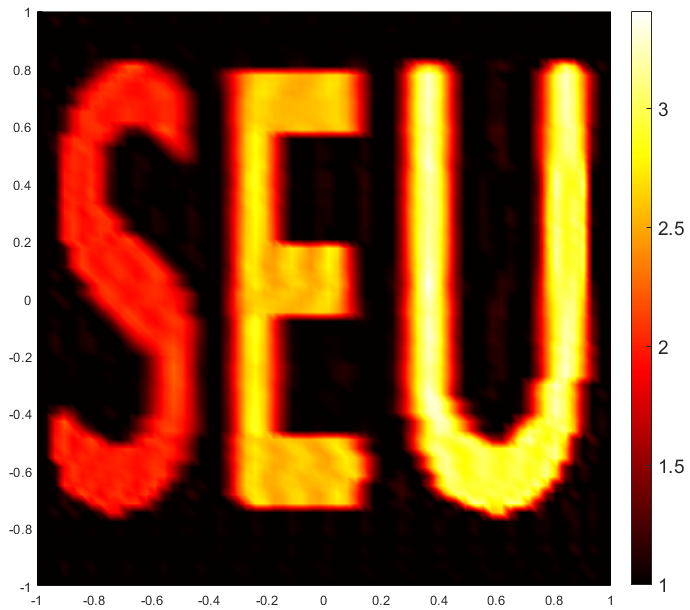}}   %image_fullmask.png
\end{minipage} \\              
\begin{minipage}[t]{0.33\linewidth}
\captionsetup{singlelinecheck=false}
\subfigure[Target conductivity]{
\includegraphics[width=6cm,height=5.5cm]{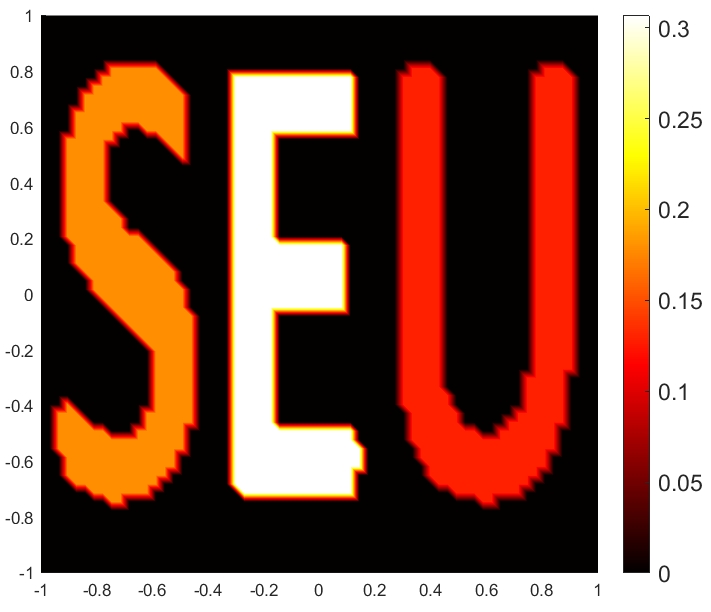}}
\end{minipage}
\begin{minipage}[t]{0.33\linewidth}
  % \centering
\subfigure[Reconstructed conductivity with $K$ = 16]
{\includegraphics[width=6cm,height=5.5cm]{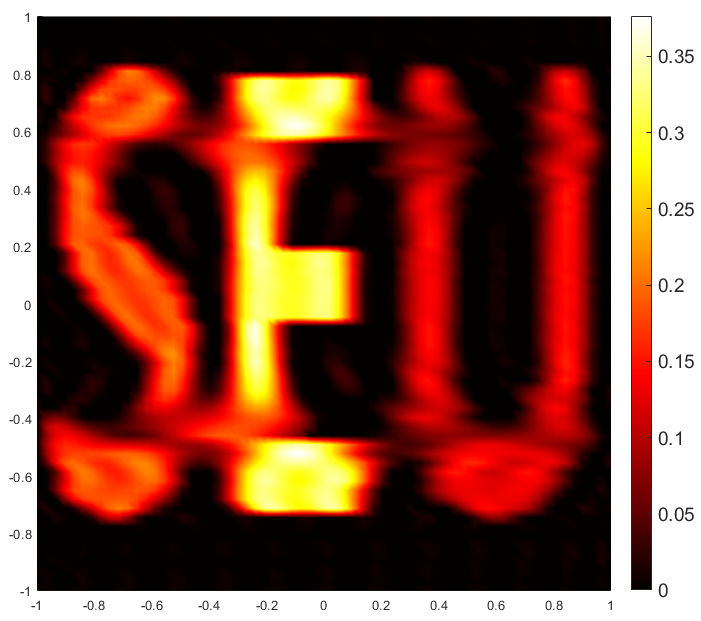}}
\end{minipage}
\begin{minipage}[t]{0.32\linewidth}
\subfigure[Reconstructed conductivity with $K$ = 64]{
\includegraphics[width=6cm,height=5.5cm]{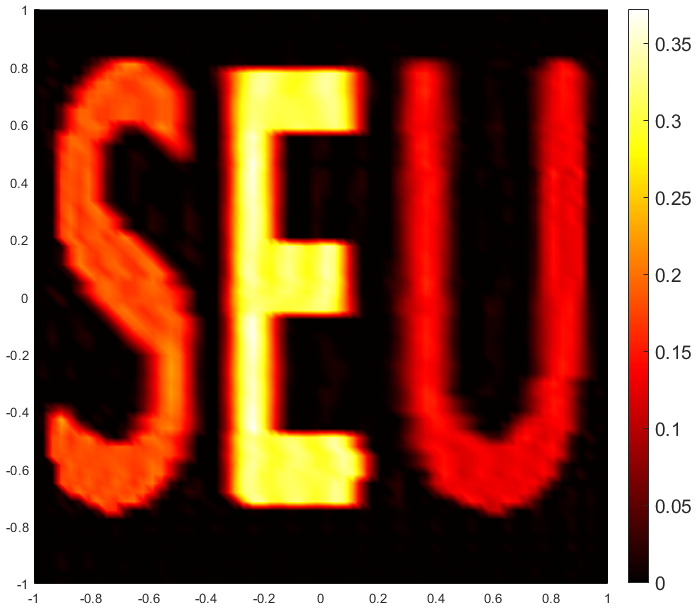}}
\end{minipage}
\caption{RPCD reconstruction results versus $K$ with $L=4$ BSs and SNR = 30 dB. 
Unit of conductivity is S/m.
}
\label{image_K}
\end{figure*}

\begin{figure}[t]
  \centering
\centerline{\includegraphics[width=8.4cm,height=6.5cm]{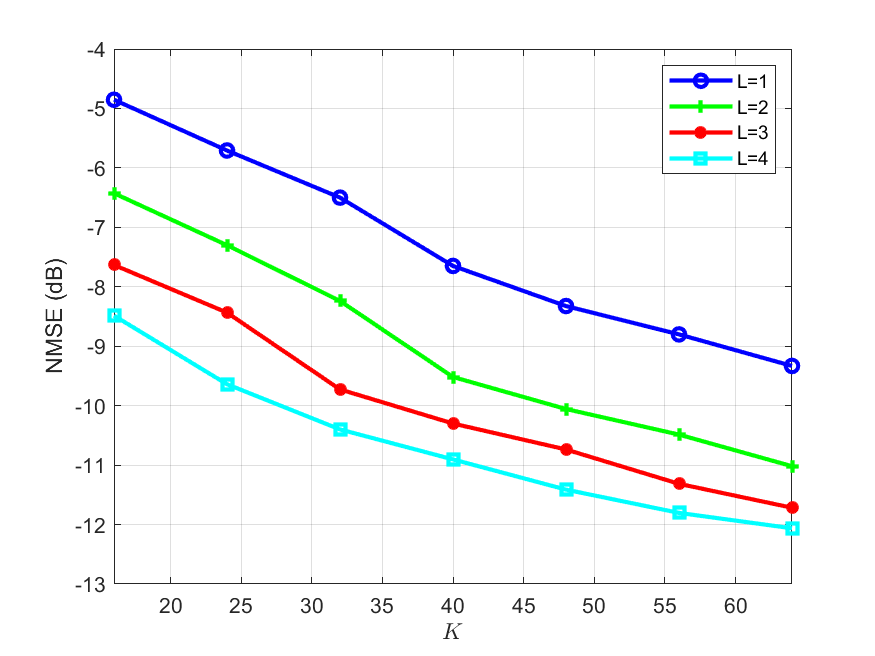}}
  \caption{\textcolor{black}{NMSE of PRCD versus $K$ with SNR = $30$ dB.}}
  \label{nmse_K}
\end{figure}

\subsection{RPCD Reconstruction Performance versus Bandwidth}
% (a) Relative permittivity and (b) conductivity of contrast obtained by the proposed method; (c) relative permittivity and (d) conductivity of contrast. 
% We set SNR = 30 dB and  

To demonstrate the RPCD reconstruction results, we present images of the target reconstructed 
based on relative permittivity and conductivity with $L = 4$ and SNR = $30$ dB, for $K=16$ and $64$, respectively. 
The target ``SEU" is the abbreviation of ``Southeast University", and the three letters are composed of wood, chipboard, and plasterboard, respectively.
As shown in Fig.~\ref{image_K}, the RPCD reconstructed with $K = 64$ is more accurate to reflect the target's shape compared to the RPCD reconstructed with $K = 16$. 
Moreover, a larger number of subcarriers yields more accurate reconstructed values of relative permittivity and conductivity, which leads to a better representation of the target's real EM property. 

We investigate the NMSE of reconstructed RPCD versus the bandwidth by changing the number of subcarriers $K$ while keeping the frequency step $\Delta f$ unchanged, as shown in Fig.~\ref{nmse_K}. 
It is seen that, NMSE decreases with the increase of $K$ in all cases, and the best performance is achieved when the signals from all 4 BSs are used. 
This is because the increased bandwidth enhances the diversity of the pilot patterns and reduces the mutual coherence of the sensing matrix.
Consequently, more information about RPCD is extracted, leading to better reconstruction accuracy. 

% As the target gets farther from the transmitter, the column mutual coherence of the sensing matrix becomes larger, causing larger errors in RPCD reconstruction. 

% Although the mutual coherence of the general sensing matrix $\tilde{\mathbf{E}}$ do not change significantly in $K\in[16,64]$, the increase of $K$ enlarges the row number of $\tilde{\mathbf{E}}$ and leads to more measurement data, which improves the accuracy of RPCD reconstruction. 

% When $K$ is smaller than $32$, the NMSE is large and decreases slowly with the increase of SNR.  
% In this stage, the differences among the NMSE for different $d$ are not pronounced. 
% The NMSE decreases rapidly when SNR increases from $15$~dB to $20$~dB.  

% When SNR reaches certain thresholds, the NMSE decreases to the error floor and hardly changes. 
% At this point, the restriction of RPCD reconstruction quality is no longer the noise at the receiver, but is the imperfect condition of the mutual coherence of the sensing matrix. 

% It is then evident that the reconstruction algorithm works well for a large number of measurements that compensate for the distortion of the generated masks.   
% When $I$ approaches the total number of pixels, the NMSE decreases sharply. This is because when $I$ is less than the total number of pixels, the number of unknown pixel values to be solved is greater than the number of equations that can be formed, making it impossible to recover the image with high quality. 

\subsection{Material Classification Accuracy versus SNR}
\begin{figure}[t]
  \centering  \centerline{\includegraphics[width=8.4cm,height=6.5cm]{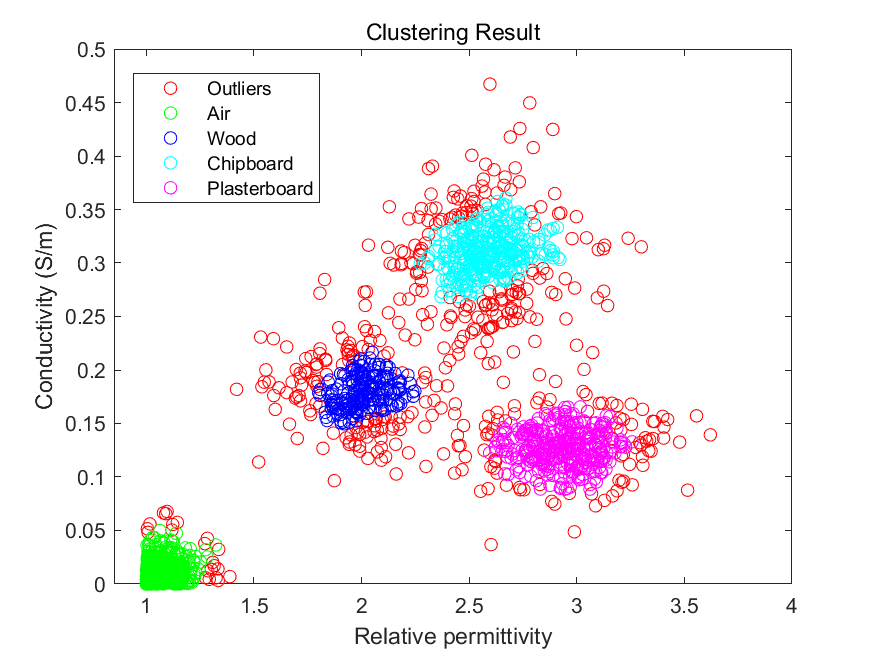}}  
  \caption{\textcolor{black}{Clustering results using the DBSCAN method.}}
  \label{dbscan}
\end{figure} 

\begin{figure}[t]
  \centering  \centerline{\includegraphics[width=8.4cm,height=6.5cm]{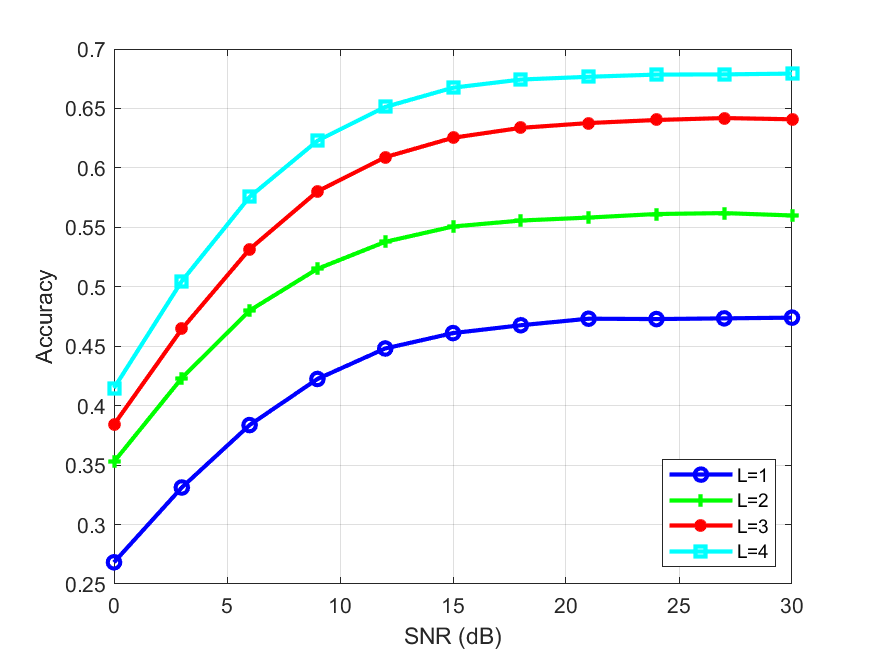}}  
  \caption{\textcolor{black}{Material classification accuracy versus SNR with $K = 32$.}}
  \label{accuracy_snr}
\end{figure}

% The Rand Index (RI) is a measure of the similarity between two data clusterings. Given a set of elements \( S \) and two partitions of \( S \) to compare, a clustering \( X \) and a clustering \( Y \), the Rand Index is calculated as follows: 
% Let:
% \begin{itemize}
%     \item \( a \) be the number of pairs of elements in \( S \) that are in the same subset in \( X \) and in the same subset in \( Y \),
%     \item \( b \) be the number of pairs of elements in \( S \) that are in different subsets in \( X \) and in different subsets in \( Y \),
%     \item \( c \) be the number of pairs of elements in \( S \) that are in the same subset in \( X \) and in different subsets in \( Y \),
%     \item \( d \) be the number of pairs of elements in \( S \) that are in different subsets in \( X \) and in the same subset in \( Y \).
% \end{itemize}
% The Rand Index is then calculated using the formula:
% \begin{align}
% RI = \frac{a + b}{a + b + c + d}
% \end{align}
% This formula represents the proportion of agreement between the two clusterings, disregarding the number of clusters used. An RI value of 1 indicates perfect agreement, while a value close to 0 indicates no agreement.
% We evaluate the accuracy of the clustering results using RI. 

Assume that the target may be composed of $10$ possible kinds of materials in the database, including wood, concrete, chipboard, etc, whose relative permittivity and conductivity have been precisely measured in advance. 
The shape of the target is set to be identical to ``THU" in Fig.~\ref{image1}, while the EM property of the target is decided by the material. 

To provide an example of classifying the materials in region $D$, 
we show the clustering results using DBSCAN in Fig.~\ref{dbscan} with $K = 32 $, SNR = $30$~dB, and $L=4$.
In Fig.~\ref{dbscan}, each data point represents a sampling point in region $D$, and the colors of the data points represent the classification results. 
It is observed that the cluster of the air is close to $(1,0)$,
whereas the clusters of the target materials are notably away from $(1,0)$. 
Therefore, this separation enables clear differentiation between sampling points associated with air and those associated with the target materials. 

Using the sampling points occupied by the target after removing the outlier points and the points associated with air,
we show the material classification accuracy versus SNR with different $L$ and $K = 32$ in Fig.~\ref{accuracy_snr}.  
It is seen that, the classification accuracy increases with the increase of SNR for all $L$, and the best performance is achieved with $L= 4$. 
When the SNR is larger than $15$~dB, the accuracies are larger than $45\%$ for all $L$.   
The accuracy increases rapidly when the SNR increases from $0$~dB to $10$~dB and reaches an upper bound at about $20$ dB. 
This is because larger SNR results in lower RPCD reconstruction errors, which offers higher classification accuracy according to Fig.~\ref{nmse1} and Fig.~\ref{accuracy_snr}. 

% We conduct 10000 Monte Carlo experiments for each material under each simulation setting.
% Generally, the larger $d$ leads to 
% the lower RPCD reconstruction error, which further results in the higher classification accuracy according to Fig.~\ref{nmse1} and Fig.~\ref{nmse2}. 

\subsection{Material Classification Accuracy versus Bandwidth} 
\begin{figure}[t]
  \centering  \centerline{\includegraphics[width=8.4cm,height=6.5cm]{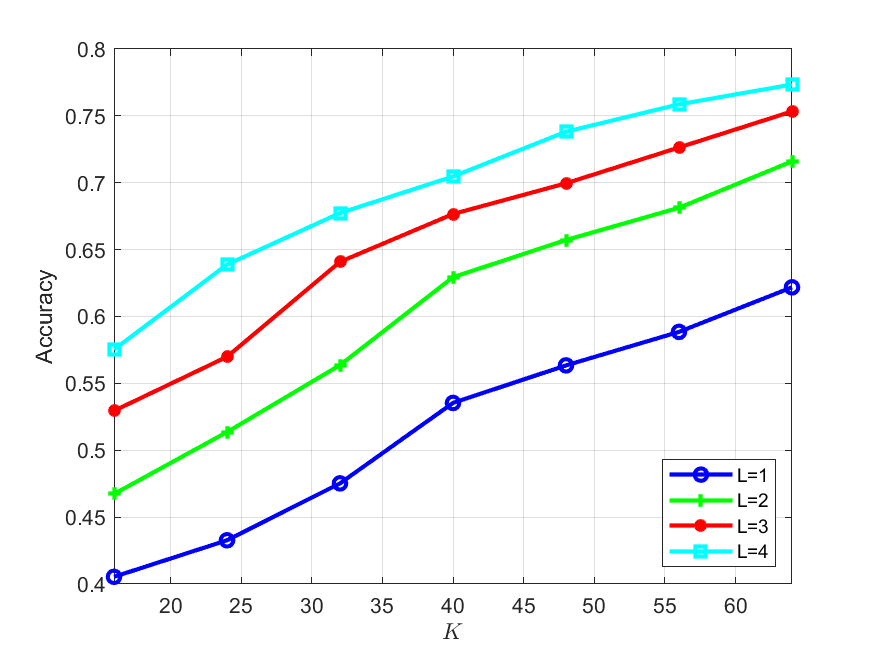}}  
  \caption{\textcolor{black}{Material classification accuracy versus $K$ with SNR = $30$ dB.}}
  \label{accuracy_K}
\end{figure}

With SNR = $30$ dB and under the same other simulation settings as in Section~\uppercase\expandafter{\romannumeral5}.C, we explore the material classification accuracy versus the bandwidth by changing the number of subcarriers $K$ while keeping the frequency step $ \Delta f $ unchanged, as shown in Fig.~\ref{accuracy_K}.  
It is seen that, the classification accuracy increases with the increase of $K$ in all cases. 
% The highest accuracy is achieved at $L = 4$ and $K=64$, where the accuracy is $88\%$.  
When $K$ is larger than $35$, the accuracies are larger than $50\%$ for all $L$. 
Generally speaking, larger $K$ leads to lower RPCD reconstruction errors, which further provides higher classification accuracy according to Fig.~\ref{nmse_K} and Fig.~\ref{accuracy_K}.

\subsection{ Influence of Channel Error }
\begin{figure}[t]
  \centering  \centerline{\includegraphics[width=8.4cm,height=6.5cm]{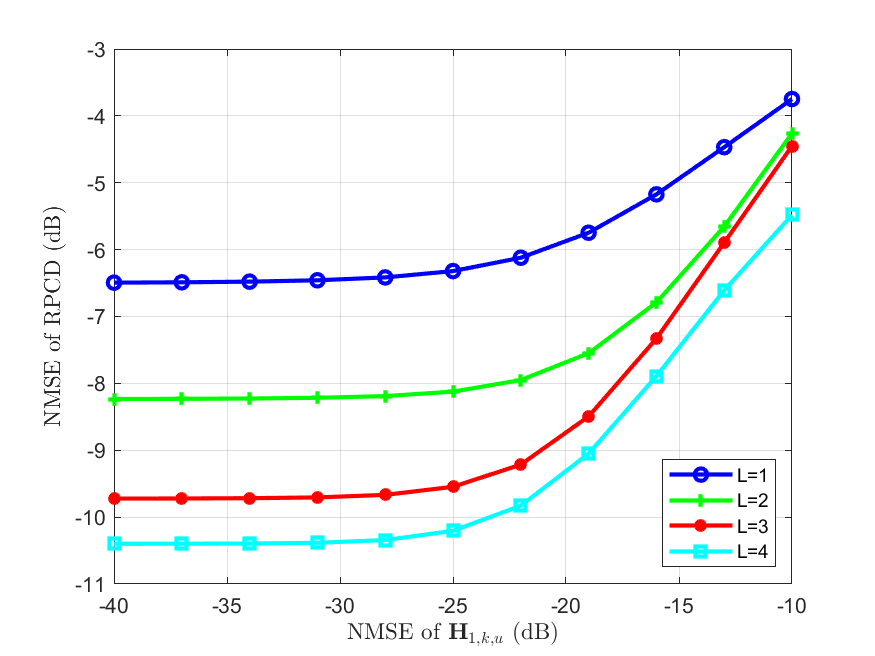}}  
  \caption{\textcolor{black}{NMSE of PRCD versus NMSE of $\mathbf{H}_{1,k,u}$ with SNR = $30$ dB and $K = 32$. }}
  \label{nmse_nmse1}
\end{figure}

\begin{figure}[t]
\vspace{-3mm}
  \centering  \centerline{\includegraphics[width=8.4cm,height=6.5cm]{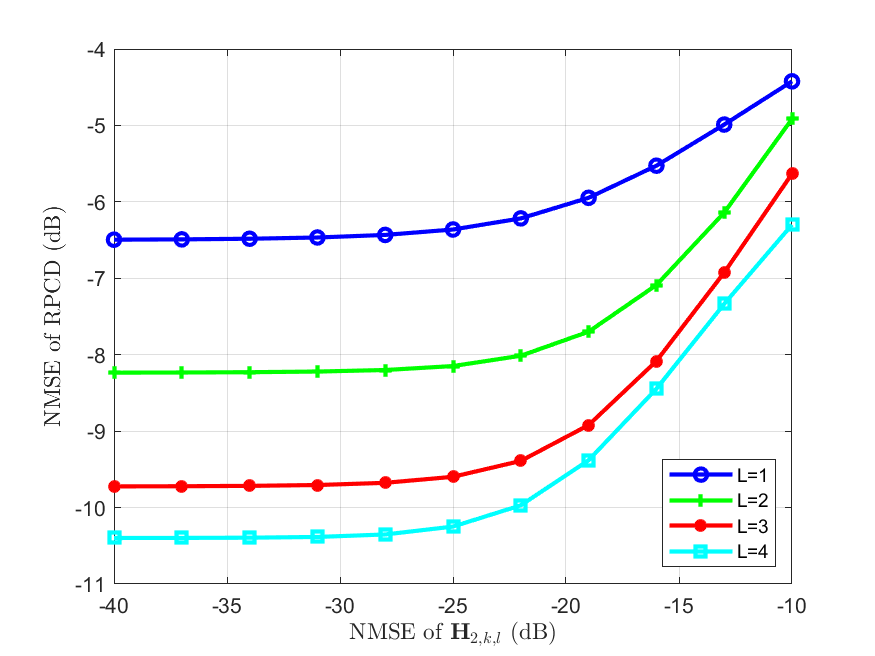}}  
  \caption{\textcolor{black}{NMSE of PRCD versus NMSE of $\mathbf{H}_{2,k,l}$ with SNR = $30$ dB and $K = 32$. }}
  \label{nmse_nmse2}
\end{figure}

We explore the influence of the channel error by adding CSCG noise to the computed channels $\mathbf{H}_{1,k,u}$ and $\mathbf{H}_{2,k,l}$.
The NMSE of PRCD versus NMSE of $\mathbf{H}_{1,k,u}$ with SNR = $30$ dB and $K = 32$ is shown in Fig.~\ref{nmse_nmse1}. 
It is seen that as the NMSE of $\mathbf{H}_{1,k,u}$ increases, the NMSE of the RPCD also rises for all values of $L$. Systems with more BSs consistently demonstrate better performance, exhibiting lower NMSE values of PRCD. 
An error floor appears when the NMSE of 
$\mathbf{H}_{1,k,u}$ is between -40 dB and -25 dB, where the NMSE curves flatten out for all values of $L$. In this range, reducing the channel error further does not significantly improve the  NMSE of PRCD. 
While the proposed method shows robustness to small channel errors of $\mathbf{H}_{1,k,u}$, increasing the number of BSs does not necessarily enhance its resilience to larger channel errors.

% \vspace{8cm}

We explore the NMSE of PRCD versus NMSE of $\mathbf{H}_{2,k,l}$ with SNR = $30$ dB and $K = 32$  in Fig.~\ref{nmse_nmse2}. 
We observe phenomena similar to Fig.~\ref{nmse_nmse1}, but with smaller NMSE values under the same NMSE of channel error.
This improvement may be attributed to the fact that the BSs are relatively farther from the target than the UEs, which causes the channel $\mathbf{H}_{2,k,l}$ to approximate a far-field single-path model.
As a result, the effective degrees of freedom for the channel $\mathbf{H}_{2,k,l}$ are smaller than $\mathbf{H}_{1,k,u}$, which makes the sensing system more resilient to the errors of $\mathbf{H}_{2,k,l}$.
The proposed method demonstrates robustness to small channel errors of $\mathbf{H}_{2,k,l}$ with varying numbers of BSs.

\subsection{ Influence of Target's Location }
\begin{figure}[t]
  \centering  \centerline{\includegraphics[width=8.4cm,height=6.5cm]{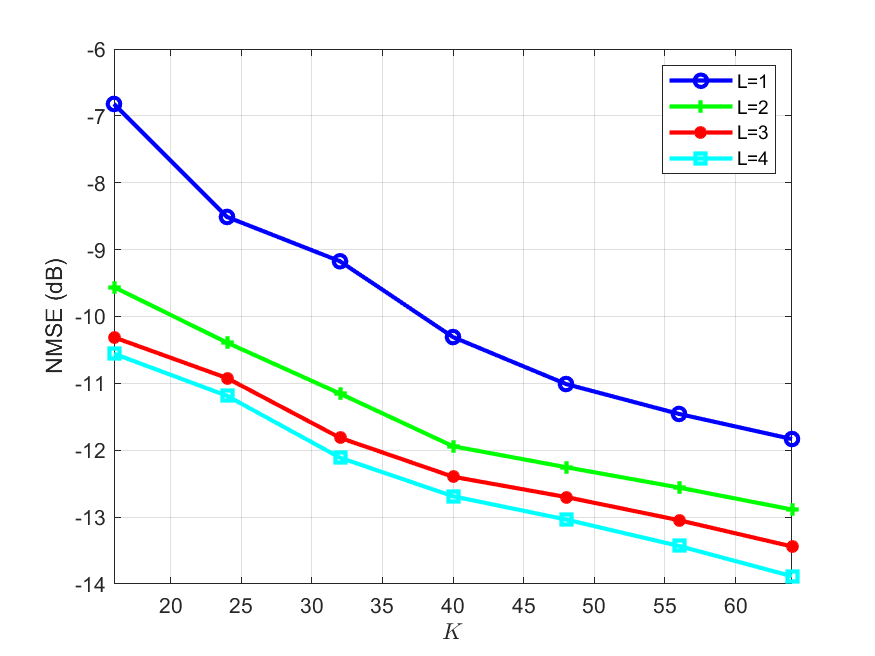}}  
  \caption{\textcolor{black}{NMSE of PRCD versus $K$ with SNR = $30$ dB and $D = [49,51] \\ \times [49,51]$~$\mathrm{m}^2$.}}
  \label{nmse_location}
\end{figure}

We explore the influence of the target's location by changing the region of the target into $D = [49,51]  \times [49,51]$~$\mathrm{m}^2$. 
The $ 10 $ UEs are randomly located in a circular region of radius $10$~m centered at $[50,50]$~m but outside region $D$.
The NMSE of PRCD versus $K$ with SNR = $30$ dB is shown in Fig.~\ref{nmse_location}, where the $4$ BSs located at $(100,0)$~m, $(0,100)$~m, $(-100,0)$~m, and $(0,-100)$~m are fused in sequence for $L=1,2,3,4$, respectively. 
We observe phenomena similar to Fig.~\ref{nmse_K} where the target region is at the center of multiple BSs, but with relatively smaller NMSE values under the same $K$.
The reason is that the BSs at $(100,0)$~m and $(0,100)$~m are closer to the target when the target is located within $D = [49,51]  \times [49,51]$~$\mathrm{m}^2$ than when the target is located within $D = [-1,1] \times [-1,1]$~$\mathrm{m}^2$.
The BSs at $(100,0)$~m and $(0,100)$~m can thus provide richer information of the target's EM property, which contributes to smaller NMSE values.
Moreover, when the target is located within $D = [49,51]  \times [49,51]$~$\mathrm{m}^2$, the BSs at $(100,0)$~m and $(0,100)$~m are closer to the target and thus have smaller NMSE of PRCD reconstruction than the BSs at $(-100,0)$~m and $(0,-100)$~m.
Adding the information collected by the BSs at $(-100,0)$~m and $(0,-100)$~m to the information collected by the BSs at $(100,0)$~m and $(0,100)$~m results in relatively small improvement in PRCD reconstruction.
Thus, the curves corresponding to $L=2,3,4$ are closer to each other with $D = [49,51]  \times [49,51]$~$\mathrm{m}^2$ than the curves corresponding to $L=2,3,4$ with $D = [-1,1] \times [-1,1]$~$\mathrm{m}^2$ in Fig.~\ref{nmse_K}.

\section{Limitations, Potential Challenges, Practical Implementation, and Future Directions}

The proposed EM property sensing framework, while effective, has several limitations.
The framework assumes static environments and is designed to sense static target.
Moreover, the multi-BS data fusion adds extra computational overhead, which may present challenges in resource-constrained environments.

Potential challenges include handling high-dimensional data and ensuring real-time processing, especially as more BSs are added. 
Accurate sensing and synchronization across multiple BSs also pose potential challenges, especially in dynamic environments with moving targets.

The practical implementation of the proposed EM property sensing framework requires both advanced hardware and software components. On the hardware side, multiple BSs must be equipped with multi-antenna arrays and have robust communication links.
Additionally, high-performance processors are necessary to handle the computational complexity of data fusion and signal processing, especially for real-time applications.
On the software side, the system requires efficient algorithms related to pilot design, compressive sensing, and data fusion for EM property sensing and material classification. Real-time operating systems or parallel computing frameworks would be essential to meet the processing speed and resource management demands.

Future directions include extending the framework to sense the EM property of moving targets, which could be beneficial in areas like surveillance and healthcare.
Incorporating machine learning techniques for enhanced material identification and exploring adaptive sensing strategies to handle dynamic conditions are promising avenues.
Moreover, improving robustness to noise and CSI errors, along with refining multi-BS fusion techniques, could significantly enhance the framework's real-world applicability.
Potential improvements in pilot design could involve designing pilots that simultaneously perform both channel estimation and EM property sensing, optimizing resource use. Fusion algorithms could be enhanced by dynamically adjusting the fusion weights across BSs, further improving performance.
Additionally, the proposed MACE method could be extended to other sensing tasks that can be formulated as the optimization of a sum of functions constructed by BSs, such as localization or velocity estimation, making it versatile for a wide range of applications.

\section{Conclusion}
In this paper, we propose an innovative scheme employing OFDM pilot signals in ISAC with multiple BSs for effective EM property sensing and material identification. 
The proposed framework not only facilitates the precise estimation of RPCD but also enhances material classification through a collaborative network approach.
Key contributions of the proposed framework include the establishment of a comprehensive EM property sensing mechanism in ISAC with multiple BSs, 
the implementation of a feature-level multi-BS data fusion strategy to optimize EM property reconstruction, 
and the strategic design of multi-subcarrier pilots to minimize column mutual coherence.
The simulation outcomes demonstrate the efficacy of the proposed framework, highlighting its capability to achieve high-quality RPCD reconstruction and accurate material classification. 
Additionally, the simulation results suggest that the reconstruction quality and classification accuracy can be significantly improved by optimizing the SNR at the receivers or by expanding the number of BSs involved in data collection. 

\color{black}

% \linespread{1.6}
 \small 
 \bibliographystyle{ieeetr}
 \bibliography{IEEEabrv,mainbib}

\ifCLASSOPTIONcaptionsoff
  \newpage
\fi 

% trigger a \newpage just before the given reference
% number - used to balance the columns on the last page
% adjust value as needed - may need to be readjusted if
% the document is modified later
%\IEEEtriggeratref{8}
% The "triggered" command can be changed if desired:
%\IEEEtriggercmd{\enlargethispage{-5in}}

% references section

% can use a bibliography generated by BibTeX as a .bbl file
% BibTeX documentation can be easily obtained at:
% http://mirror.ctan.org/biblio/bibtex/contrib/doc/
% The IEEEtran BibTeX style support page is at:
% http://www.michaelshell.org/tex/ieeetran/bibtex/
%\bibliographystyle{IEEEtran}
% argument is your BibTeX string definitions and bibliography database(s)
%\bibliography{IEEEabrv,../bib/paper}
%
% <OR> manually copy in the resultant .bbl file
% set second argument of \begin to the number of references
% (used to reserve space for the reference number labels box)

%  The first group is generated with the most clusters, and subsequent groups added to the test dataset are generated with increa\singly scarce clusters. The growth rates of both curves decrease with the increase of group numbers, because the newly added $\mathbf{H}$ can generally be compressed by larger CR.
\end{document}